\documentclass[12pt]{article}
\usepackage{amsmath}
\usepackage{graphicx,psfrag,epsf}
\usepackage{enumerate}
\usepackage{natbib}
\usepackage{color}
\usepackage{longtable}
\usepackage{url} 

\allowdisplaybreaks
\usepackage{multirow}
\newcommand{\blind}{0}

\graphicspath{{Fig/}}

\usepackage{array,booktabs}
\usepackage{amsbsy,amssymb,epsf,here,amsmath,amsfonts,latexsym,amsthm}
\usepackage{caption}
\usepackage{subfigure}
\usepackage{tikz}
\usepackage{footmisc}

\usepackage[mathscr,mathcal]{euscript}
\usepackage{natbib}
\usepackage{float}
\usepackage{lscape}
\usepackage{pdflscape} 
\usepackage{longtable}
\usepackage{threeparttable}
\usepackage{threeparttablex}

\usepackage{cite} 
\usepackage{xr}
\usepackage{url}
\usepackage{chapterbib}
\usepackage{docmute}
\usepackage{enumitem}
\usepackage{hyperref}
\newcommand*{\addFileDependency}[1]{
	\typeout{(#1)}
	\@addtofilelist{#1}
	\IfFileExists{#1}{}{\typeout{No file #1.}}
}
\makeatother

\newtheorem{theorem}{Theorem}
\newtheorem{proposition}[theorem]{Proposition}%

\newtheorem{assumption}{Assumption}

\def\bo10z{\mbox{\bo10math $Z$}}
\def\lbo10z{\mbox{\bo10math $z$}}
\def\bo10y{\mbox{\bo10math $Y$}}
\def\lbo10y{\mbox{\bo10math $y$}}
\def\bo10u{\mbox{\bo10math $u$}}
\def\E{\mathbb{E}}

\def\full{\mathrm{full}}
\def\obs{\mathrm{obs}}
\def\pr{\textnormal{pr}}

\def\P{\mathbb{P}}
\def\bo10d{\mbox{\bo10math $d$}}

\newcommand{\MR}{\textnormal{tr}}

\newcommand\Pn{{\mathbb{P}_n}}
\def\T{{ \mathrm{\scriptscriptstyle T} }}
\makeatletter
\newcommand*{\indep}{%
	\mathbin{%
		\mathpalette{\@indep}{}%
	}%
}
\newcommand*{\nindep}{%
	\mathbin{
		\mathpalette{\@indep}{\not}
	}%
}
\newcommand*{\@indep}[2]{%
	\sbox0{$#1\perp\m@th$}
	\sbox2{$#1=$}
	\sbox4{$#1\vcenter{}$}
	\rlap{\copy0}
	\dimen@=\dimexpr\ht2-\ht4-.2pt\relax
	\kern\dimen@
	{#2}%
	\kern\dimen@
	\copy0 
}

\makeatother

\begin{document}

	\def\spacingset#1{\renewcommand{\baselinestretch}%
		{#1}\small\normalsize} \spacingset{1}

	\newcounter{savecntr}
\newcounter{restorecntr}
	
	\if0\blind
	{
		\title{\bf Multiply robust estimation of causal  effects using linked data}
		\author{Shanshan Luo\setcounter{savecntr}{\value{footnote}}\thanks{Both authors contributed equally to this work.}
  \textsuperscript{1},     Yechi Zhang\setcounter{restorecntr}{\value{footnote}}%
  \setcounter{footnote}{\value{savecntr}}\footnotemark
  \setcounter{footnote}{\value{restorecntr}} \textsuperscript{2}, and  
			Wei Li\textsuperscript{2} \\\\
			\textsuperscript{1} School of Mathematics and Statistics, \\Beijing Technology and Business University\\\\
			\textsuperscript{2} Center for Applied Statistics and School of Statistics, \\Renmin University of China\\
 }
 
		\date{}
		
		\maketitle
	} \fi
	
	\if1\blind
	{
		\bigskip
		\bigskip
		\bigskip
		\begin{center}
			{\LARGE\bf Title}
		\end{center}
		\medskip
	} \fi
	
	\bigskip
	\begin{abstract}
		Unmeasured confounding presents a common challenge in observational studies, potentially making standard causal parameters unidentifiable without additional assumptions. 
		Given the increasing availability of diverse data sources, exploiting data linkage offers a potential solution to mitigate unmeasured confounding within a primary study of interest. 
		However, this approach often introduces selection bias, as data linkage is feasible only for a subset of the study population.  To address this concern, we explore three nonparametric identification strategies under the assumption that a unit's inclusion in the linked cohort is determined solely by the observed confounders, while acknowledging that the ignorability assumption may depend on some partially unobserved covariates. 
		The existence of multiple identification strategies {motivates} the development of estimators that effectively capture distinct {components} of the observed data distribution.  Appropriately combining these estimators yields triply robust estimators for the average treatment effect. These estimators remain consistent if at least one of the three distinct parts of the observed data law is correct. Moreover, they are locally efficient if all the  models are correctly specified.  We evaluate the proposed  estimators using simulation studies and real data analysis.
	\end{abstract}
	
	\noindent%
	{\it Keywords:} Causal inference;  Linked data;  Missing confounder; Two-phase sampling.
	\vfill
	
	\newpage
	\spacingset{1.8} 
\section{Introduction}
	\label{sec:intro}
	
	Unmeasured confounding remains a persistent challenge within observational studies, leading to {biased} estimations of causal parameters. In the current era of big data, the increasing availability of diverse  data sources offers potential remedies. Among these, leveraging data linkage emerges as a  {promising approach} to mitigate the impact of unmeasured confounding in a primary study of interest. For instance, in healthcare research, the linkage of a claims database from a health plan with an electronic health record database from a delivery system can yield richer patient data. The resulting linked cohort, comprising patients present in both data sources, presents an  opportunity to enhance estimation by incorporating pivotal confounding factors.
	However,  the data linkage approach may  introduce selection bias. This arises from the fact that studies conducted within linked databases are often  {restricted} to a subset of the primary study population. If this subset fails to adequately represent the entire population, constraining analysis solely to the linked data might 
	{invalidate} findings and {limit} result generalizability \citep{sun2022use}.

	The problem of integrating inferences from primary and supplementary datasets to control for unmeasured confounding has previously received substantial attention, for example, within the context of  two-phase sampling designs \citep{chatterjee2003pseudoscore,wang2009causal}. When more detailed confounding information becomes available from an alternative study, 
	\citet{sturmer2005adjusting} employed the regression calibration technique to mitigate estimation bias in a large healthcare database with incomplete confounder data.
	Similarly, in a comparable scenario,  \citet{mccandless2012adjustment} presented a Bayesian framework for adjusting missing confounders  through the incorporation of external data and propensity scores. Furthermore, \citet{lin2014adjustment}  developed a two-stage calibration method that effectively integrates propensity scores to {summarize} confounding information, ultimately merging results from both the main and validation datasets.
	However, these methods primarily focus on regression parameters that might not perfectly align with the causal parameters of interest, additionally necessitating validation dataset representative of the primary   population. {In a similar spirit  to the  calibration approach  by  \citet{lin2014adjustment}}, \citet{yang2019combining}   introduced a general framework for estimation of causal effects combining the primary dataset with unmeasured confounders and a smaller  validation dataset that contributes  additional confounding information.
	For scenarios involving heterogeneity between the two data sources, \citet{sun2022use} proposed an estimator employing {two} weightings to estimate the causal effect. Within this scheme, the inverse probability of treatment weights was utilized to deal with confounding bias, while inverse probability of selection weights was applied to address selection bias within the linked cohort. However,
	the estimators proposed by \citet{yang2019combining} and \citet{sun2022use} may not be the most efficient, particularly in situations with heterogeneity between the two datasets.

	This paper focuses on the semiparametric efficient estimation of the average treatment effect (ATE) using a linked database that may exhibit heterogeneity from the primary   population. We establish three nonparametric identification formulas and subsequently develop three 
	corresponding 
	semiparametric   estimators. These estimators operate under the assumption of ignorable treatment assignment which allows for dependence on some partially unobserved confounders,  and the assumption 
	that a unit's inclusion in the linked data hinges solely on observed covariates.
	The presence of multiple estimators for the same causal estimand suggests the potential for their combination  to enhance estimation. We  derive the efficient influence function (EIF) for ATE within a nonparametric model and introduce a semiparametric estimator demonstrating the triply robust property \citep{bickel1993efficient}. This estimator remains consistent when considering the union of three distinct components of  likelihood and achieves local efficiency if all these models are correctly specified. Furthermore, our proposed estimator retains its consistency and asymptotic normality even when leveraging flexible models derived from machine learning approaches \citep{chernozhukov2018double}.   
	These findings extend the concept of doubly robust estimation for ATE within observational studies featuring complete-case analysis \citep{bang2005doubly} and align with the  {principles} of triply robust estimators in other causal inference domains \citep{wang2018bounded, Jiang2022multiply}. 
	
	Moreover, we {explore}  the application of the proposed estimator to identify optimal sampling designs by minimizing the {asymptotic} variance under specific cost constraints. Additionally, we consider  semiparametric efficient estimations of the causal risk ratio for binary outcomes. We also discuss the connection between the proposed method and existing approaches for handling missing data, wherein individuals outside the linked cohort can be viewed as those with missing values for certain confounding variables. Despite various methods  proposed for estimating causal effects  using  complete data, the efficient estimation of causal parameters in the presence of missing data has received less attention. \citet{seaman2014inverse} investigated the use of multiple imputation for handling missing confounders, followed by employing inverse-probability weighting to estimate causal quantities under the missing at random (MAR) assumption. \citet{mayer2020doubly} further introduced doubly robust causal effect estimation with missing confounders. \citet{williamson2012doubly} addressed scenarios involving missing outcomes, missing exposures, and single missing confounders separately, proposing ``multiply robust" estimators of treatment effects. However, as  \citet{evans2020coherent} demonstrated, the method proposed by \citet{williamson2012doubly} fails to be multiply robust due to variation dependence between outcome and confounder imputation models. Additionally, \citet{evans2020coherent}  explored estimating the average treatment effect among the treated in situations where  one confounder is MAR. They developed a doubly robust approach grounded in variation independent components, although 
	this approach is complicated due to the necessity of modeling conditional odds ratios. In contrast, we present a novel semiparametric efficient estimation framework for ATE under a slightly strengthened assumption compared to MAR, which motivates the development of
	a class of triply robust estimators based on variation independent nuisance models.
	
	The rest of this paper is organized as follows.  Section \ref{sec: notation} introduces the notation, assumptions,  and  three different identification formulas.    In Section \ref{sec: estimation}, we derive the EIF for ATE and propose the triply robust estimators. Section \ref{sec:extension} discusses potential extensions of our proposed method and explores its connections with missing data. We evaluate the finite-sample performance of the proposed method  through simulations and a real data example in Section \ref{sec:numerical}. Finally, Section \ref{sec:concludes} concludes  with a brief discussion. 
	
	\section{Notation, assumptions, and identification}
	\label{sec: notation}
	
	Suppose we are interested in estimating the ATE of a binary treatment $Z$ on an outcome $Y$ within a primary study of interest.
	Following the framework outlined in \citet{rubin1974estimating}, we adopt potential outcomes to define causal effects and make the stable unit treatment value assumption, wherein each unit has a single version of potential outcomes and no interference occurs between units. Let $Y_z$ denote  the potential outcome that would be observed if the treatment $Z$ were $z$ for $z=0,1$. Under the causal consistency assumption, the observed outcome can be expressed as $Y=ZY_1+(1-Z)Y_0$. The ATE is denoted as $\tau=\E(Y_1-Y_0)$. In addition to  the treatment   $Z$ and   outcome   $Y$, a set of observed covariates denoted by $X$ is also observed in the primary dataset.  
Simultaneously, an available supplementary dataset provides additional covariates  represented by $V$, although it lacks the other variables present in the primary dataset. We use $R$ to denote a binary indicator variable which takes the value 1 if a unit appears in both datasets and 0 if the unit is exclusive to the primary dataset. The collective inclusion of units observed in both datasets, typically a subset of the primary data, constitutes what we refer to as the linked cohort data.   We hence have complete information on all relevant variables including $(V, X, Z, Y)$ within this linked data.
	Under a superpopulation model, we assume that the observations in the combined dataset
	$\{O_i = (R_i, R_iV_i, X_i, Z_i, Y_i),i=1,\ldots,n\} $ are  identically and independently distributed. To account for potential heterogeneity between the two data sources, we introduce the following assumption:


	
	\begin{assumption}
		\label{ass: MAR-V}
		$R \indep (Y,Z,V)\mid X$.    
	\end{assumption}
	
	Assumption \ref{ass: MAR-V} means that the probability of being in the linked cohort depends solely on fully observed covariates.
	This assumption was implicitly considered in addressing potential selection bias within linked cohort data by \citet{sun2022use}. A weaker version of Assumption \ref{ass: MAR-V}, known as the MAR assumption, is denoted as $R \indep V \mid (X,Y,Z)$. This MAR assumption implies that the selection mechanism is predicated on all the observed data  \citep{robins1994estimation}. While the MAR assumption is inherently less stringent than our Assumption \ref{ass: MAR-V}, we can verify the latter by examining conditional independence between $R$ and $(Z,Y)$ given $X$   within scenarios that comply with MAR conditions.
	
	
	A fundamental problem {for} causal inference is that we can never observe two potential outcomes for a unit simultaneously. Following \citet{rosenbaum1983central},  we make the following ignorability assumption  to identify ATE. 
	\begin{assumption}\label{ass:ignorability}
		$Z\indep(Y_0,Y_1)\mid(X,V)$.
	\end{assumption}     
	Assumption \ref{ass:ignorability} posits that the treatment assignment depends on both    $X$ and  $V$. When  {combined}  with Assumption \ref{ass: MAR-V}, this implies that the treatment assignment is ignorable in the linked cohort with $R=1$ given the observed variables   $X$ and  $V$. However, in the subgroup with $R=0$, the treatment assignment is only ``latent" ignorable given  $X$ and the latent variable $V$.  In the subsequent discussion, we define the following terms: the conditional probability of treatment assignment, often referred to as the propensity score, denoted as $\pi(X,V) =\pr(Z=1\mid X,V)$; the conditional probability of the selection indicator, known as the selection probability, denoted as ${\rho}(X) =\pr(R=1\mid X)$; and the conditional outcome mean, denoted as $\mu(Z,X,V) =\E(Y\mid Z ,X,V)$. Additionally, we introduce   $f(V\mid X)$ to denote the conditional density of the partially  observed covariate  $V$   given the fully  observed covariate $X$. This is often referred to as the imputation model. Under Assumptions~\ref{ass: MAR-V} and \ref{ass:ignorability}, we derive that $\pi(X,V)=\pr(Z=1\mid X,V,R=1)$, $f(V\mid X)=f(V\mid X,R=1)$, and $\mu(Z,X,V) =\E(Y\mid Z,X,V,R=1)$. These estimands can be consistently estimated using the linked data. Next, we impose the positivity assumption on treatment assignment, which states that for all levels of  $X$ and $V$, there are both treated and control units present.
	\begin{assumption} \label{ass:positivity-pi}
		$0<\pi(X,V)<1$ for all $X$ and $V$. 
	\end{assumption} 
	We also make the assumption that the selection probability $\rho(X)$ for being observed in the linked cohort is bounded away from $0$. This is formalized by the following assumption. 
	\begin{assumption} \label{ass:positivity-e}
		$ \rho(X)>0$ for all $X$.  
	\end{assumption}When $\rho(X)=1$, it corresponds to the scenario of complete-case analysis without selection bias, where each unit in the primary dataset can be supplemented with the additional confounding factor $V$.  In cases where $\rho(X) < 1$, Assumptions \ref{ass: MAR-V}-\ref{ass:positivity-e} still suffice for identifying the desired causal quantity.  This observation  is summarized in the following theorem.
	
	\begin{theorem}\label{THEOREM:IDENTIFIABLE}
		Under Assumptions~\ref{ass: MAR-V}-\ref{ass:positivity-e}, ATE is identifiable through {the }following formulas.
		\begin{itemize}
			\item[(a)] Based on  selection probability and propensity score,  	\begin{align*}
				\tau =\E\bigg\{ \frac{R}{\rho(X)}\frac{ Z}{ \pi(X,V )} Y \bigg\}-\E\bigg\{ \frac{R}{\rho(X)}\frac{1-Z}{1-\pi(X,V )}Y \bigg\} .
			\end{align*}
			
			\item[(b)] Based on selection probability and outcome mean,
			\begin{align*}   \tau =\E\left[\frac{R}{
					\rho(X)} \{\mu(1,X,V  )-\mu(0,X,V )\}\right]. 
			\end{align*}
			\item[(c)] Based on   imputation model and outcome mean,  	\begin{align*}
				\tau=\E\{\delta(1,X)- \delta(0,X)  \},
			\end{align*} 
			where $\delta(z,X)=\E\{\mu(z,X,V)\mid X\}$.
		\end{itemize}
	\end{theorem}  
	Theorem \ref{THEOREM:IDENTIFIABLE} establishes the identification of ATE through three distinct components of observed likelihood. Below we offer further insights into these three strategies for identification.   In Theorem \ref{THEOREM:IDENTIFIABLE}(a), the ATE is expressed as the difference between weighted averages of the outcome under different treatment assignments \citep{sun2022use}. This formulation involves two sets of weights: $R/\rho(X)$, which addresses the selection bias; $Z/\pi(X,V)$ and $(1-Z)/\{1-\pi(X,V)\}$, which address the confounding bias.
	The identification expression in Theorem \ref{THEOREM:IDENTIFIABLE}(a) can be interpreted as an analogy to inverse probability weighting for identifying ATE while using another inverse selection probability weighting to account for selection bias.

	Theorem \ref{THEOREM:IDENTIFIABLE}(b)  provides the identification of ATE through  the selection probability and    outcome mean. In comparison to Theorem \ref{THEOREM:IDENTIFIABLE}(a), two terms involving the propensity score, $ZY/\pi(X,V)$ and $(1-Z)Y/\{1-\pi(X,V)\}$, are replaced by the outcome means $\mu(1,X,V)$ and $\mu(0,X,V)$, respectively. Much like Theorem \ref{THEOREM:IDENTIFIABLE}(a),
	this expression can be   interpreted  as an outcome regression analogy for identifying ATE while using inverse probability weighting to adjust for selection bias.

	Theorem \ref{THEOREM:IDENTIFIABLE}(c) presents an expression for ATE in terms of the outcome mean and the imputation model. Under Assumptions~\ref{ass: MAR-V}-\ref{ass:positivity-e}, Theorem \ref{THEOREM:IDENTIFIABLE}(c) identifies ATE by marginalizing the aforementioned conditional expectation $\delta (1,X) - \delta(0,X)$ over the distribution of $X$.
	
	\section{Estimation}
	\label{sec: estimation}
	\subsection{Semiparametric estimation}\label{subsec:semi}
	In this section, we  describe several estimation methods for ATE. The three identification formulas presented in Theorem \ref{THEOREM:IDENTIFIABLE} motivate three estimators, whose consistency and asymptotic normality rely on the correct specification of different nuisance functions. For simplicity, we use $\mathbb{P}_n$ to represent the empirical mean operator, defined as $\mathbb{P}_n (U)=n^{-1}\sum_{i=1}^n U_i$ for a generic random variable $U$. 
	We use working models $\rho(X;\alpha)$, $\pi(X,V;\beta)$, $\mu(Z,X,V;\gamma)$, and $f(V\mid X;\theta)$ for the selection probability $\rho(X)$, propensity score $\pi(X,V)$, outcome mean $\mu(Z,X,V)$, and imputation model $f(V\mid X)$, respectively.  We can obtain a consistent estimator $\widehat{\alpha}$ through maximum likelihood estimation (MLE). Additionally, consistent estimators $\widehat{\beta}$, $\widehat{\gamma}$, and $\widehat{\theta}$ can be constructed through MLE or moment methods based on the linked data.  Let $\alpha^\star$, $\beta^\star$, $\gamma^\star$, and $\theta^\star$ denote the probability limits of $\widehat{\alpha}$, $\widehat{\beta}$, $\widehat{\gamma}$, and $\widehat{\theta}$, respectively.  To motivate and clarify our proposed triply robust estimator, we introduce three classes of semiparametric models. These models impose parametric restrictions on different components of the observed data likelihood, while allowing the remaining models to be unrestricted.

	
	\begin{enumerate}[leftmargin=15pt]
		\item[$\mathcal{M}_1$:] The selection probability   ${\rho}(X;\alpha)$ and {the }propensity score  $\pi(X,V;\beta)$ are correctly specified, such that $\rho(X;\alpha^\star)=\rho(X)$ and $\pi(X,V;\beta^\star)=\pi(X,V)$.
		\item[$\mathcal{M}_2$:] The selection probability   ${\rho}(X;\alpha)$ and {the} outcome mean  $\mu(Z,X,V;\gamma)$   are correctly specified, such that $\rho(X;\alpha^\star)=\rho(X)$ and $\mu(Z,X,V;\gamma^\star)=\mu(Z,X,V )$.
		\item[$\mathcal{M}_3$:] The outcome mean $ \mu(Z,X,V;\gamma)  $ and {the} imputation model $f(V\mid X;\theta)$ are correctly specified,  such that  $\mu(Z,X,V;\gamma^\star)=\mu(Z,X,V )$ and $f(V\mid X;\theta^\star)=f(V\mid X)$.
	\end{enumerate}
	
	The identification formula  in Theorem \ref{THEOREM:IDENTIFIABLE}(a) motivates the following  estimator based on     selection probability and propensity score: 
	\begin{equation*}
			\begin{aligned}
				\widehat\tau_{1}=\Pn\bigg\{\frac{R}{{\rho}(X;\widehat\alpha)}\frac{ZY}{\pi(X,V;\widehat\beta)}  \bigg\}-\Pn\bigg\{ \frac{R}{{\rho}(X;\widehat\alpha)}\frac{(1-Z)Y}{1-\pi(X,V;\widehat\beta)} \bigg\}.
			\end{aligned} 
	\end{equation*} 
	This estimator employs two inverse probability functions, which is similar to the Horvitz-Thompson type estimator \citep{horvitz1952generalization}. However, weighting-based estimators can be unstable in practice if the estimated propensity score is  close to zero or one. To address this issue, the Hájek-type weighting estimator is introduced, aiming to reduce variance compared to the Horvitz-Thompson estimator in practice \citep{hajek1971comment}.  The Hájek-type  weighting estimator    is given by	
	\begin{equation*} 
			\begin{aligned}
				&\widehat\tau_{1}^\prime=\Pn\bigg\{\frac{R}{{\rho}(X;\widehat\alpha)}\frac{ZY}{\pi(X,V;\widehat\beta)}  \bigg\}\bigg/\Pn\bigg\{\frac{R}{{\rho}(X;\widehat\alpha)}\frac{Z }{\pi(X,V;\widehat\beta)}  \bigg\}\\&~~~~~~~- \Pn\bigg\{ \frac{R}{{\rho}(X;\widehat\alpha)}\frac{(1-Z)Y}{1-\pi(X,V;\widehat\beta)} \bigg\}\bigg/\Pn\bigg\{\frac{R}{{\rho}(X;\widehat\alpha)}\frac{1-Z }{1-\pi(X,V;\widehat\beta)}  \bigg\}.
			\end{aligned}   
	\end{equation*}  
	Under standard regularity conditions, both weighting estimators $\widehat{\tau}_1$ and  $\widehat{\tau}_1^\prime$  are consistent and asymptotically normal within the submodel $\mathcal{M}_1$.  
	
	The identification formula  in Theorem \ref{THEOREM:IDENTIFIABLE}(b) motivates the following   estimator  based on  selection probability and outcome mean model:
	\begin{align*}   \widehat\tau_2 =\Pn\left[\frac{R\{\mu(1,X,V ;\widehat{\gamma} )-\mu(0,X,V;\widehat{\gamma} )\}}{
			{\rho}(X;\widehat\alpha )} \right]. 
	\end{align*}
	Similar to the estimator based on the  selection probability and propensity score,
	we can construct the stabilized   version of $\widehat\tau_2$ as follows:  
	\begin{align*}   \widehat\tau_2^\prime =\Pn\left[\frac{R\{\mu(1,X,V ;\widehat{\gamma} )-\mu(0,X,V;\widehat{\gamma} )\}}{
			{\rho}(X;\widehat\alpha )} \right]\bigg/\Pn\left\{\frac{R}{
			{\rho}(X;\widehat\alpha )}  \right\}. 
	\end{align*} 
	Under regularity conditions, both   estimators $\widehat{\tau}_2$ and  $\widehat{\tau}_2^\prime$ are  consistent and asymptotically normal  within the  submodel $\mathcal{M}_2$.   
	
	The identification formula  in Theorem \ref{THEOREM:IDENTIFIABLE}(c) motivates the following   estimator  based on the outcome mean  and the imputation model:	
	$$\widehat\tau_3 = \Pn\big\{\delta(1,X;\widehat{\gamma},\widehat{\theta}) -\delta(0,X;\widehat{\gamma},\widehat{\theta}) \big\},$$ 
	where
	$\delta(z,X;\widehat{\gamma},\widehat{\theta})=D^{-1}\textstyle\sum_{d=1}^D\mu(z,X,\widetilde V_d ;\widehat{\gamma})$
	and $\{\widetilde V_d :d=1,\ldots,D\}$ are $D$  observations independently drawn from the imputation  model $f(V\mid X;\widehat{\theta})$. 
	Previous studies have explored various imputation techniques to address missing covariates  \citep{Mitra2011EstimatingPS,seaman2014inverse}. 
	The main idea for $\widehat\tau_3$ involves  imputing the missing covariate $V$ for the primary dataset and integrating it into the outcome mean model to estimate ATE.  
	Within the submodel $\mathcal{M}_3$, the  estimator $\widehat\tau_3$ is consistent and asymptotically normal under regularity conditions. 
	\subsection{Triply robust estimation}
	In the preceding section, we demonstrated that the estimation of ATE can be achieved through three different formulas, allowing analysts to construct a variety of estimators by combining these approaches. However, this flexibility also calls for the development of more principled estimation techniques. Moreover, in the presence of high-dimensional covariates, it is challenging to ascertain the correct specification of   these nuisance  models. Consequently, it becomes crucial to develop robust estimators that can handle potential model misspecifications. In this section, our objective is to construct a consistent and asymptotically normal estimator under the union model $\mathcal M_{\mathrm{union}}=\cup_{j=1}^3\mathcal{M}_j$. To achieve this, we first derive the efficient influence function  under the nonparametric model of the observed data distribution, which serves as the foundation for our triply robust estimator.  
	\begin{theorem}
		\label{THM: EIF}
		Under Assumptions~\ref{ass: MAR-V}-\ref{ass:positivity-e},
		the efficient influence function for ATE is: $$\psi_{\mathrm{eff}}^\tau(O)=\varphi_1(O)-\varphi_0(O)-\tau,$$ 
		where  $ \varphi_{ z}(O)$ is defined as follows:  
		\begin{align}
			\label{eq: phiz-exp}
			\dfrac R{\rho(X)}\bigg[
			\dfrac {\mathbb{I}(Z=z)\left\{Y-\mu(z,X,V)\right\}}{\{\pi(X,V)\}^z\{1-\pi(X,V)\}^{1-z}}
			+\mu(z,X,V)-\delta(z,X) \bigg]   +  \delta(z,X).   
		\end{align}   
	\end{theorem}
	According to Theorem \ref{THM: EIF}, the semiparametric efficiency bound for  ATE under Assumptions~\ref{ass: MAR-V}-\ref{ass:positivity-e} is given  by $\E\{\psi_{\mathrm{eff}}^\tau(O)\}^2$. 
	Building upon this insight, we can enhance estimation efficiency compared to the estimators outlined in Section~\ref{subsec:semi}.  We now proceed to construct a locally efficient estimator and demonstrate its consistency and asymptotic normality within  $\mathcal M_{\mathrm{union}}$. Theorem \ref{THM: EIF} motivates  the following estimator:   
	\begin{equation}
		\label{eq: mr-estimator}
		\begin{aligned}
			\widehat\tau_{\MR}&= \Pn\{\widehat\varphi_1(O)- \widehat\varphi_0(O)\},
		\end{aligned}
	\end{equation} 
	where   $   \widehat  \varphi_{ z}(O)$ is defined as follows:
	\begin{equation}
		\label{eq:est-expre} 
			\begin{aligned}
				\frac R{\rho(X;\widehat{\alpha})}\Bigg[
				\dfrac {\mathbb{I}(Z=z)\left\{Y-\mu(z,X,V;\widehat\gamma)\right\}}{\{\pi(X,V;\widehat\beta)\}^z\{1-\pi(X,V;\widehat\beta)\}^{1-z}}
				+\mu(z,X,V;\widehat\gamma)-\delta(z,X;\widehat{\gamma},\widehat{\theta})
				\Bigg]   +  \delta(z,X;\widehat{\gamma},\widehat{\theta}).   
			\end{aligned}  
	\end{equation}   
	Although   $\widehat{\tau}_{\MR}$ involves models for   selection probability, propensity score, outcome mean and imputation model, its consistency and  asymptotic normality  do  not require the correct specification of all nuisance models. The following theorem   summarizes the triply robust property of $\widehat{\tau}_{\MR}$.
	\begin{theorem}
		\label{THM: MUL-ROBUSTNESS}
		Under standard regularity conditions, $\widehat{\tau}_{\MR}$ is  consistent and asymptotically normal {under} the union model $\mathcal{M}_{\mathrm{union}}$. Moreover, if all {the }models in $\mathcal{M}_1$, $\mathcal{M}_2$, and $\mathcal{M}_3$ are correctly specified, $\widehat{\tau}_{\MR}$ attains the semiparametric efficiency bound.
	\end{theorem}
	The asymptotic variance formula of $\widehat{\tau}_\MR$   can  be derived using standard M-estimation theory  \citep{newey1994large}. Alternatively, we can also use the nonparametric bootstrap  for variance estimation in practice.  Now
	we offer intuitive explanations for the triple robustness property.  Specifically,  
	the probability limit of $\Pn\{\widehat\varphi_1(O)\}-\E(Y_1)$ can be expressed as $B_{1 }+B_{2}+B_{3} $, where 
	\begin{equation*}
			\begin{aligned}    
				B_{  1 }&=\mathbb{E}\left[\begin{matrix}
					\dfrac{ {\mu}\left(1,X,V;\gamma^\star\right)R}{f (V\mid X ) {{\rho}(X)}{{\rho}(X;\alpha^\star)}}\left\{ { {\rho}(X;\alpha^\star)}-{{\rho}(X)}
					\right\}\times \left\{f (V\mid X )- {f} (V\mid X;\theta ^\star)\right\}
				\end{matrix}\right],\\
				B_{2}& =\mathbb{E}\left[\begin{matrix}
					\dfrac{\pi (X,V) R}{ \pi (X,V;\beta^\star)   {{\rho}(X)}{ {\rho}(X;\alpha^\star)}}\left\{ { {\rho}(X;\alpha^\star)}-{{\rho}(X)}
					\right\}\times\left\{\mu(1,X,V)- {\mu}(1,X,V;\theta^\star)\right\}
				\end{matrix}\right],\\B_{3}&=  \mathbb{E}\left[\begin{matrix}
					\dfrac R{{ \pi (X,V;\beta^\star) }\rho(X)}\left\{\pi(X,V)-\pi(X,V;\beta^\star)\right\}\times\left\{\mu\left(1,X,V\right)- {\mu}(1,X,V;\theta^\star)\right\}
				\end{matrix}\right].
			\end{aligned}
	\end{equation*} 
	The  technical details  are provided in the supplementary material.   As a sanity check, it can be verified that the bias term $B_1$ equals zero when employing the correct selection probability or imputation model; the bias term $B_2$ equals zero when utilizing the correct selection probability or outcome mean model;  and the bias term $B_3$ equals zero when using the correct propensity score or outcome mean model. Therefore, the asymptotic bias converges to zero under the union model, indicating the triple robustness property of $\widehat{\tau}_{\MR}$.

	We  now examine the properties of $\widehat{\tau}_{\MR}$ when nonparametric estimation is used for nuisance functions. Importantly, we want to emphasize that $\widehat{\alpha},\widehat{\beta},\widehat{\gamma }$ and $\widehat{\theta}$ can also represent nonparametric nuisance function estimates, such as sieves, splines, random forests, or neural networks. By incorporating more flexible estimation approaches for the nuisance functions, we can mitigate the potential risk of parametric model misspecification in  practice.
	Given certain regularity conditions, the bias formulas mentioned earlier offer an alternative avenue to achieve consistency and asymptotic normality when utilizing flexible estimates for the nuisance functions. This concept is similar to the doubly debiased machine learning approach used for estimating ATE in complete-case analysis \citep{chernozhukov2018double}. We define  the $L_2$-norm of a 
	function $h$ with respect to the distribution $\nu$ of a generic variable $U$  by $\Vert h\Vert_2=\{\int h^2(u)\nu(du)\}^{1/2}$.
	
	\begin{theorem}
		\label{THM: NON-MARCHINE-LEARNING}	Suppose that Assumptions~\ref{ass: MAR-V}-\ref{ass:positivity-e} hold, and we further introduce the following conditions:
		\begin{itemize} \item[(a)]  {the functions} $\{\rho(X;\widehat{\alpha}) $, $\pi(X,V;\widehat\beta)$, $\mu(Z,X,V;\widehat \gamma)$, $f(V\mid X;\widehat \theta)\}$ converge in probability to $\{\rho(X),\pi(X,V), \mu(Z,X,V)$, $f(V\mid X)\}$,  for all  $X$ and $V$;
			\item[(b)] {the} functions  $ \{\rho(X;\widehat{\alpha})$, $\pi(X,V;\widehat\beta)$, $\mu(Z,X,V;\widehat \gamma)$, $f(V\mid X;\widehat \theta)\}$ and  $\{\rho(X)$, $\pi(X,V)$, $\mu(Z,X,V)$, $f(V\mid X)\}$ are in a Donsker class;  
			\item[(c)]  {for all  $X$ and $V$,  there exist   constants  $M>0$ and $0<\epsilon <1$ satisfy{ing} that:  
				\begin{gather*}
					-M<\left\{   \mu(Z,X,V;\widehat \gamma) ,       \mu(Z,X,V)    \right\}<M; \\
					\epsilon<\left\{\rho(X;\widehat{\alpha}), \pi(X,V;\widehat\beta), \rho(X),  \pi(X,V) \right\}<1-\epsilon .
				\end{gather*} 
				We also assume that $f(V\mid X;\widehat \theta)>\kappa$  and $f(V\mid X)   >\kappa$ for all  $X$ and $V$, where $\kappa$ is a positive constant.} 
			\item[(d)]The following three product terms  are   of order   $o_p(n^{-1/2})$,
			\begin{equation}    \label{eq: bias-terms}
				\begin{gathered}
					{\big\|\rho(X)-\rho(X;\widehat\alpha )\big\|}_2{\big\|m(X,V)-m(X,V ;\widehat\gamma )\big\|}_2 ,\\ {\big\|\rho(X)-\rho(X;\widehat\alpha )\big\|}_2{\big\|f (V\mid X )-f (V\mid X ;\widehat\theta  )\big\|}_2 ,\\ {{\big\|\pi(X,V)-\pi(X,V ;\widehat\beta )\big\|}_2} {\big\|m(X,V)-m(X,V ;\widehat\gamma )\big\|}_2,
				\end{gathered}
			\end{equation}
		\end{itemize}where $m(X,V)=\mu(1,X,V)-\mu(0,X,V)$ and $m(X,V;\widehat\gamma )=\mu(1,X,V;\widehat\gamma  )-\mu(0,X,V; \widehat\gamma )$.  
		Then $\widehat\tau_{\MR}$ in \eqref{eq: mr-estimator} is   asymptotically normal, has the influence function $ {\psi_{\mathrm{eff}}^\tau}(O)$, and achieves the semiparametric efficiency bound.
	\end{theorem} 
	The consistency requirement in condition (a) has been extensively examined in various widely used flexible models  \citep{farrell2021deep}. Condition (b) introduces limitations on the complexity of the spaces including the nuisance functions and their corresponding estimators \citep{kennedy2016semiparametric,van2000asymptotic}. To  alleviate  the need for these empirical process conditions, techniques such as sample splitting and cross-fitting can be employed \citep{chernozhukov2018double,newey2018cross}.  The bounded constraints outlined in condition (c) suffice to constrain the error $    \Pn\{\widehat\varphi_1(O)- \widehat\varphi_0(O)\}-\tau$ by  the summation of the terms in a form analogous to  \eqref{eq: bias-terms}. Similar to condition (a), 
	condition (d) is a modest requirement and can be satisfied in many scenarios. For instance, as long as the estimates of the nuisance functions converge
	faster than $n^{-1/4}$, condition (d) is satisfied, and this can be accomplished by flexible models under sparsity or smoothness constraints.
	\section{Extensions and connections}
	\label{sec:extension}
	\subsection{Design issue: optimal two-phase sampling ratio} 
	In this subsection, we focus on tackling the design issue of   determining the optimal sample allocation that minimizes the {asymptotic} variance of the proposed estimator. This design  problem shares similarities with the classical two-phase sampling approach \citep{cochran2007sampling}. We introduce  a more stringent assumption than Assumption \ref{ass: MAR-V}, specifically assuming that the linked data is a simple random sample from the primary study population. 
	\begin{assumption}
		\label{eq: indep-R}
		$R\indep (Z,X,V,Y)$.
	\end{assumption}
	{Assumption \ref{eq: indep-R} was also considered by \citet{yang2019combining} in the context of two-phase sampling. Under this assumption,} the selection probability $\rho(X)$ is a constant value,  denoted simply as $\rho$. Consider that it costs $C_1$ to collect $(Z, X, Y)$ for each unit and $C_2$ to collect $V$ for each unit. The total cost of the study can then be expressed as follows:
	\begin{align}
		\label{eq:cost}
		C = n (C_1 + \rho C_2).
	\end{align}
	The expression of the EIF for ATE under Assumption \ref{eq: indep-R} remains identical to \eqref{eq: phiz-exp}, {albeit} with the selection probability $\rho(X)$ substituted by the constant value $\rho$. Given that the selection probability $\rho(X)$ remains constant and can always be accurately specified, the estimator $\widehat\tau_{\MR}$ maintains consistency if either the outcome mean or the propensity score is accurate (i.e., if submodel $\mathcal{M}_1$ or $\mathcal{M}_2$ holds), as indicated by Theorem \ref{THM: MUL-ROBUSTNESS}. The correct specification of the imputation model $f(V\mid X;\theta)$ in $\mathcal{M}_3$ does not impact the consistency of $\widehat\tau_{\MR}$. This suggests that under Assumptions~\ref{ass:ignorability}--\ref{eq: indep-R}, $\widehat\tau_{\MR}$ transforms into a doubly robust estimator. However, even though the accurate specification of the imputation model $f(V\mid X;\theta)$ is not pivotal for ensuring consistency, it still can enhance estimation efficiency.

	To emphasize that $\widehat{\tau}_{\MR}$ depends solely on the parameters $(\widehat{\beta}, \widehat{\gamma}, \widehat{\theta})$, we denote the proposed estimator as $\widehat{\tau}_{\MR}(\widehat{\beta}, \widehat{\gamma}, \widehat{\theta})$ in this section. We summarize the influence function  for $\widehat{\tau}_{\MR}(\widehat{\beta}, \widehat{\gamma}, \widehat{\theta})$ in the following proposition, and the  detailed proofs are provided in the supplementary material.
	
	\begin{proposition}\label{lem:eif-tau-mr}
		Under Assumptions~\ref{ass:ignorability}-\ref{eq: indep-R},    the influence function of  $ \widehat{\tau}_{\MR} ( \widehat{\beta}, \widehat{\gamma}, \widehat{\theta}) $ is given by
		\begin{align*}
			\begin{aligned}
				\psi_{\MR}^{\tau}(O)
				& ={S_1(\beta^\star, \gamma^\star, \theta^\star)}+\frac{R}{\rho} S_2(\beta^\star, \gamma^\star, \theta^\star),
			\end{aligned}
		\end{align*} 
		where    $S_1(\beta,\gamma,\theta)  =  \delta(1,X;\gamma,\theta)-\delta(0,X;\gamma,\theta)-\tau$ and the explicit expression of $S_2(\beta,\gamma,\theta)$ is provided in the supplementary material due to its complexity.
	\end{proposition} 
	According to Proposition~\ref{lem:eif-tau-mr}, the asymptotic variance of the proposed estimator $\widehat{\tau}_\MR( \widehat{\beta}, \widehat{\gamma}, \widehat{\theta})$ under the constraint \eqref{eq:cost} can be calculated as follows:
	\begin{equation}
		\label{eq:var-tau-rand} 
			\begin{aligned}
				\mathrm{asyvar}\{\widehat\tau_\MR( \widehat{\beta}, \widehat{\gamma}, \widehat{\theta})\} &=\frac{C_2\Gamma_1}{C} \rho+ \frac{C_1\Gamma_2}{C}\frac{1}{\rho}+\frac{C_1\Gamma_1+C_2 \Gamma_2}{C}, 
			\end{aligned} 
	\end{equation}
	where $\Gamma_1=\E\left\{S_1^2(\beta^\star, \gamma^\star, \theta^\star)\right\} +2 \E\{ S_2(\beta^\star, \gamma^\star, \theta^\star) S_1(\beta^\star, \gamma^\star, \theta^\star)\mid R=1\} $, and  $\Gamma_2=\E\{S_2^2(\beta^\star,\\ \gamma^\star, \theta^\star)\mid R=1\}$. 
	By minimizing \eqref{eq:var-tau-rand} with respect to $\rho$, we obtain the optimal values $n^\star$ and $\rho^\star$, satisfying $n^\star= {C}/{(C_1+\rho^\star C_2)}$ and
	\begin{equation}
		\label{eq:local-opt}  
		\begin{aligned}
			\rho^\star= \begin{cases}\left(\dfrac{ {\Gamma_2 C_1}}{{\Gamma_1 C_2}} \right)^{1/2} & \text { if } \Gamma_2 C_1<\Gamma_1 C_2, \\ ~~~~~1, & \text { if } \Gamma_2 C_1\geq \Gamma_1 C_2.\end{cases}\end{aligned}
	\end{equation}
	We provide a graphical illustration of the solution  \eqref{eq:local-opt}  in Figure \ref{fig:subfigures}. The expression in \eqref{eq:local-opt} indicates that the optimal sample ratio  $\rho^\star$ is inversely proportional to the square root of the costs. Moreover,  $\rho^\star$ is also influenced by the variances of $S_1(\beta^\star, \gamma^\star, \theta^\star)$ and $S_2(\beta^\star, \gamma^\star, \theta^\star)$ as well as their covariance.   
	Under the same constraint \eqref{eq:cost}, \citet{yang2019combining} also explored the optimal allocation problem under Assumption \ref{eq: indep-R}. While their proposed error-prone estimator achieves doubly robustness in certain scenarios, it may not be the most efficient and hence may not yield the lowest estimation variance when all nuisance models are correctly specified. In contrast, the optimal sample size and sampling ratio recommended in \eqref{eq:local-opt} would result in the lowest estimation variance. 
	
	\begin{figure*}[!t]%
		\centering
		\resizebox{0.95\textwidth}{!}{
			\subfigure[\small $\Gamma_1 > 0$ and $\Gamma_2 C_1 < \Gamma_1 C_2$]{
				\begin{tikzpicture}
					\draw[-stealth] (0,0) -- (5,0) node[right] {$\rho$};
					\draw[-stealth] (0,-0.5) -- (0,4) node[above] { };
					
					\draw[domain=0.25:3.5,smooth,variable=\x,blue] plot ({0+\x},{0.25+1*\x+0.5/\x}) node[right] {$\mathrm{asyvar}(\widehat\tau_\MR)$};
					
					\draw (1.25,0) -- (1.25,-0.15) node[below] {$1$};         
					\draw[dashed] (0.707,0) -- (0.707,4) node[above] {$\rho^\star$};
					
				\end{tikzpicture}
			}
			\hfill
			\subfigure[\small $\Gamma_1 > 0$ and $\Gamma_2 C_1 \geq \Gamma_1 C_2$]{
				
				\begin{tikzpicture}
					\draw[-stealth] (0,0) -- (5,0) node[right] {$\rho$};
					\draw[-stealth] (0,-0.5) -- (0,4) node[above] { };
					
					\draw[domain=0.7:3.5,smooth,variable=\x,blue] plot ({0+\x},{-3.+1*\x+4/\x}) node[right] {$\mathrm{asyvar}(\widehat\tau_\MR)$};
					
					\draw (1.25,0) -- (1.25,-0.15) node[below] {$1$};             
					\draw[dashed] (1.25 ,0) -- (1.25 ,4) node[above] {$\rho^\star$};
					
				\end{tikzpicture}
			}
			\hfill
			\subfigure[\small $\Gamma_1 \leq 0$]{ 
				\begin{tikzpicture}
					\draw[-stealth] (0,0) -- (5,0) node[right] {$\rho$};
					\draw[-stealth] (0,-0.5) -- (0,4) node[above] { };
					
					\draw[domain=0.35:1.2,smooth,variable=\x,blue] plot ({0.5+\x},{1-2*\x+1/\x}) node[right] {$\mathrm{asyvar}(\widehat\tau_\MR)$};
					
					\draw (1.25,0) -- (1.25,-0.15) node[below] {$1$}; 
					\draw[dashed] (1.25 ,0) -- (1.25 ,4) node[above] {$\rho^\star$};
				\end{tikzpicture}
		}}
		\caption{The optimal    second-phase sample ratio $\rho^\star$ under different scenarios.}
		\label{fig:subfigures}
	\end{figure*}
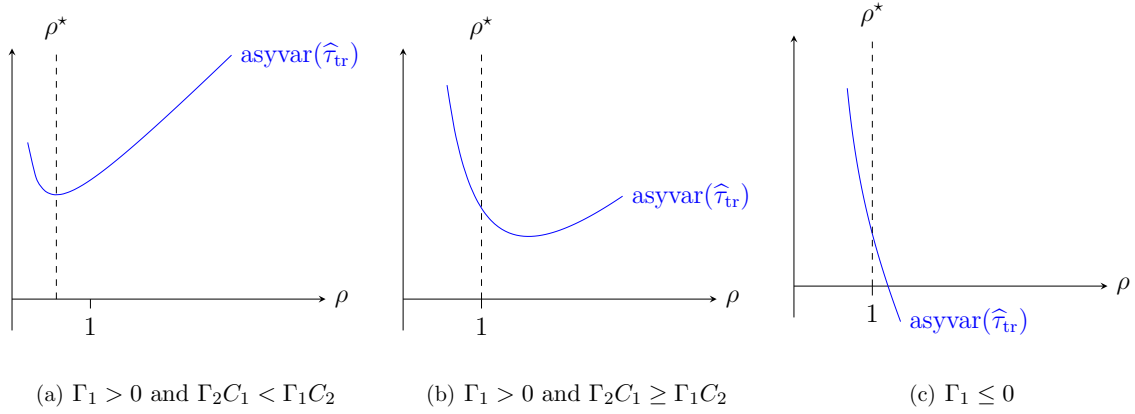
	
	\subsection{Other causal estimands}
	
	The proposed methodology can be extended to study the identification and semiparametric estimation of various  causal estimands based on the linked data. In this section, we focus on exploring the causal risk ratio (CRR) within the context of a binary outcome. Let $\xi$ denote the CRR, which represents the ratio of the average potential outcome under treatment   to that under control; that is, $\xi = {\mathbb{E}(Y_1)}/{\mathbb{E}(Y_0)}$.
	Parallel to Theorem \ref{THEOREM:IDENTIFIABLE}, we establish the identification results of  $\xi$ under Assumptions~\ref{ass: MAR-V}-\ref{ass:positivity-e}.
	\begin{proposition}
		\label{PROP:CRR-IDENTIF}
		Under Assumptions~\ref{ass: MAR-V}-\ref{ass:positivity-e}, the CRR is identifiable through  the  following formulas.
		\begin{itemize}
			\item[(a)] Based on  selection probability and propensity score,  	\begin{align*}
				\xi =\E\bigg\{ \frac{R}{\rho(X)}\frac{ Z}{ \pi(X,V )} Y \bigg\}\bigg/\E\bigg\{ \frac{R}{\rho(X)}\frac{1-Z}{1-\pi(X,V )}Y \bigg\} .
			\end{align*}
			
			\item[(b)] Based on selection probability and outcome mean,
			\begin{align*}   \xi =\E\left[\frac{R}{
					\rho(X)} \{\mu(1,X,V  ) \}\right]\bigg/\E\left[\frac{R}{
					\rho(X)} \{ \mu(0,X,V )\}\right]. 
			\end{align*}
			\item[(c)] Based on   imputation model and outcome mean,  	\begin{align*}
				\xi=\E\big\{\delta(1,X) \big\} \big/\E\big\{ \delta(0,X )\big  \}.
			\end{align*} 
		\end{itemize}
	\end{proposition} 
	
	The identification expressions for    CRR  are similar to those for the ATE in Theorem \ref{THEOREM:IDENTIFIABLE},  {albeit} with a transformation from additive scales to multiplicative ones. 
	These three identification formulations  can be  utilized to develop estimators based on three distinct parametric models, adopting a similar methodology as applied to the equations for  $\widehat\tau_1,\widehat\tau_2$, and $\widehat\tau_3$.  Additionally, we present  the EIF  for  causal risk ratio $\xi$ in the  following theorem.
	\begin{theorem}
		\label{THM: MTYPE-EIF}
		The efficient influence function for $\xi$  is given by 
		\begin{align*}  
			\varphi_{\mathrm{eff}}^\xi(O)&=   \{{\varphi_{1}(O)-{ \xi}\varphi_{0} (O)}\}/{\tau_{0}}  ,
		\end{align*} 
		where $\varphi_z(O)$ is defined   earlier in   \eqref{eq: phiz-exp}, and $\tau_0=\E(Y_0)$.
	\end{theorem}  
	
	We utilize the EIF  $\varphi_{\mathrm{eff}}^\xi(O)$ to estimate the causal risk ratio, incorporating estimates of  nuisance parameters.  Theorem \ref{THM: MTYPE-EIF} motivates  the following  plug-in estimator:  {\begin{equation*} 
			\begin{aligned}
				\widehat \xi_{\MR}&=\dfrac{\Pn\{{\widehat\varphi_{1}(O) } \}-\widehat\xi\Pn\{{ \widehat\varphi_{0} }(O)\}}{{\Pn \{\widehat\varphi_{0} (O)\}}}  +\widehat\xi 
				=\dfrac{\Pn\{{\widehat\varphi_{1}(O) } \} }{{\Pn \{\widehat\varphi_{0} (O)\}}}   
				.
			\end{aligned}
		\end{equation*}  
		The above equation    implies that the final form of $\widehat{\xi}_{\MR}$ remains unchanged,   regardless of the choice of the initial plug-in estimator  $\widehat\xi$, as long as the efficient estimator $\Pn\{\widehat{\varphi}_{0}(O)\}$ is used for estimating $\tau_0$. The estimator $\widehat{\xi}_{\MR}$ can be seen as a direct sample mean ratio, employing the efficient influence function $\varphi_z(O)$. Thus,
		the intrinsic triply robust property of $\widehat{\xi}_{\MR}$ naturally holds. Furthermore, $\widehat{\xi}_{\MR}$  attains the semiparametric efficiency bound under the intersection submodel $\cap_{j=1}^3\mathcal{M}_j$.  We summarize these findings in the following theorem.
	}
	
	\begin{theorem}
		\label{THM: MTYPE-MUL-ROBUSTNESS}
		Under standard regularity conditions, $\widehat{\xi}_\MR$ is  consistent and asymptotically normal {under} the union model $\mathcal{M}_{\mathrm{union}}$. Moreover, if all {the }models in $\mathcal{M}_1$, $\mathcal{M}_2$ and $\mathcal{M}_3$ are correctly specified, $\widehat{\xi}_\MR$ attains the semiparametric efficiency bound.
	\end{theorem}  
	It is important to note that although $\widehat \xi_{\MR}$ incorporates several nuisance models, achieving its consistency and asymptotic normality does not necessitate the accurate specification of all nuisance functions. The estimator $\widehat{\xi}_\MR$ remains robust in delivering  consistent estimates, even in the presence of  some misspecified nuisance functions.  
	
	\subsection{Connection with missing data} 
	\label{sec:missing} 
	In this section, we discuss connections between our framework and existing methods in missing data analysis, focusing on the context of covariates missing at random. The MAR assumption  in our setting is stated as follows:
	
	\begin{assumption}
		\label{assumption: MAR2}
		$R\indep V\mid (Z,X,Y)$. 
	\end{assumption}
	
	With a slight abuse of notation, we proceed to utilize similar notations for the nuisance functions, but now {tailored to fit} within the framework of  Assumption \ref{assumption: MAR2}. For instance, we define the conditional probability of selection as $\rho(Z,X,Y) = \pr(R=1 \mid Z,X,Y)$ and the density of $V$ given $(Z,X,Y)$ as $f(V \mid Z,X, Y)$. 
	Accordingly, we parameterize the corresponding  nuisance models using $\rho(Z,X,Y;\alpha)$ and $f(V \mid Z,X, Y;\theta)$ as done in the previous section.   

 Currently, there is no existing proposal in the literature that achieves multiply robust estimation within the context of covariates missing at random.  
	\citet{williamson2012doubly} introduced a unified approach for estimating causal effects using a multiply robust methodology under Assumption  \ref{assumption: MAR2}. They argued that their estimator would remain consistent if (i) either the  $\pi(X,V;\beta)$  or the  $\mu(Z,X,V;\gamma)$ is correct, and (ii) either the $f(V\mid Z,X,Y;\theta)$ or the $\rho(Z,X,Y;\alpha)$ is correct.   This implies a potential quadruple robustness property for their proposed estimator.   However, as elucidated in \citet{evans2020coherent}, achieving multiple robustness entails placing parametric restrictions on specific components of the observed data likelihood while leaving others unrestricted. Consequently,  the multiply robust property claimed by \citet{williamson2012doubly} might not be attainable due to the lack of variation independence among the parametric models for the different components of the observed likelihood.  In the context of estimating  average treatment effect among the treated, 
 \citet{evans2020coherent} proposed a method of parametrizing the likelihood to ensure variation independence among the components of the observed likelihood, and demonstrated that their estimator achieves double robustness. Their method is based on modeling conditional odds ratios \citep{yun2007semiparametric}, which can pose challenges in terms of accurate specification and estimation in practice. 

	As discussed below Assumption  \ref{ass: MAR-V}, while Assumption  \ref{assumption: MAR2} is less restrictive, its validity enables us to employ observed data to test  Assumption  \ref{ass: MAR-V}. In this paper, we have introduced a genuinely multiply robust estimator for ATE under Assumption~\ref{ass: MAR-V}. The nuisance models employed in this paper are variational independent and   commonly used in causal inference for robust estimation. Particularly, in the scenario of complete-case  analysis with $\rho(X)=1$, the propensity score   $\pi(X,V;\beta)$ and the outcome regression   $\mu(Z,X,V;\gamma)$ are frequently utilized. One doubly robust approach to estimating   ATE  is through the  augmented inverse-propensity weighting  estimator \citep{bang2005doubly}, which is given by:  
	\begin{equation*} 
			\begin{aligned}
				\Pn\begin{bmatrix}
					{\dfrac {Z\left\{Y-\mu(1,X,V;\widehat\gamma)\right\}}{\pi(X,V;\widehat\beta)}+\mu(1,X,V;\widehat\gamma)  }   
					-{\dfrac{(1-Z)\left\{Y-\mu(0,X,V;\widehat\gamma)\right\}}{1-\pi(X,V;\widehat\beta)}+\mu(0,X,V;\widehat\gamma )} \end{bmatrix}.
			\end{aligned}  
	\end{equation*}


	
	\section{Numerical studies}\label{sec:numerical}
	\subsection{Simulation}
	\label{sec: simulation}
	In this section, we evaluate the finite-sample performance of the proposed estimators under various settings. We consider the following data-generating mechanism:
	\begin{itemize} 
		\item[(a)] Fully observed covariate: $X\sim N(0,1)$.
		\item[(b)] Selection mechanism: $\pr(R=1\mid X)=\mathrm{expit}(0.75 + 0.5 X )$, where $\mathrm{expit}(u)=\exp(u)/\{1+\exp(u)\}.$
		\item[(c)] Partially observed covariate: $V\mid X \sim N(0.5+ 0.5X,  1 )$.
		\item[(d)] Treatment assignment: $\pr(Z=1\mid X,V)=\mathrm{expit}(0.5+0.5X+0.6V)$.
		\item[(e)] Observed outcome: $Y \mid Z,X,V\sim N(0.5+0.5X+0.5V+2Z+2ZX+ZV,1)$. 
	\end{itemize}
	We are interested in estimating the average treatment effect $\tau$, with a true value of $2.5$. We assess the finite-sample properties of four semiparametric estimators: $\widehat{\tau}_1$, $\widehat{\tau}_2$, $\widehat{\tau}_3$, and the triply robust estimator $\widehat{\tau}_{\MR}$. Our analysis includes sample sizes of $1000$, $5000$, and $10000$. To evaluate the performance of these estimators in scenarios with model misspecification, we adopt an approach similar to that of \citet{kang2007demystifying}. We introduce transformed variables, denoted as $X^{\dagger}= {\left|X\right|}^{1/2}$ and $V^{\dagger}= {\left|V\right|}^{1/2}$, and examine how the estimators perform when using $X^\dagger$ and $V^\dagger$ in place of $X$ and $V$ within the working models.
	More specifically, we present the results from the following five scenarios:
	\begin{itemize}
		\item[(i)] All the models are correct.
		\item[(ii)] Only models $\rho(X;\alpha)$ and $\pi(X,V;\beta)$ in $\mathcal{M}_1$  are correct.
		\item[(iii)] Only models $\rho(X;\alpha)$ and $\mu(Z,X,V;\gamma)$ in $\mathcal{M}_2$  are correct.
		\item[(iv)] Only models   $\mu(Z,X,V;\gamma)$  and   $f( V\mid X;\theta)$  in $\mathcal{M}_3$  are correct.
		\item[(v)]  None of the models are correct.
	\end{itemize}

	Table \ref{tab: MC-prop} summarizes the performance of the four estimators   $\widehat{\tau}_1$,  $\widehat{\tau}_2$,  $\widehat{\tau}_3$ and  $\widehat\tau_{\MR}$  under the above five scenarios. The simulation results are based on 200 Monte Carlo runs and include metrics of mean bias, empirical standard error, and 95\% coverage probability. For clarity,  the bias, empirical standard deviation, and the 95\% coverage probability are all multiplied by $10^2$. Across all sample sizes, the four estimators perform well in scenario (i). As expected, the triply robust estimator $\widehat\tau_{\MR}$ exhibits small bias in the first four scenarios (i)-(iv). In contrast, the other three  estimators $\widehat{\tau}_1$, $\widehat{\tau}_2$, $\widehat{\tau}_3$   show substantial bias when their respective models are misspecified.  For example, the  weightings-based estimator $\widehat{\tau}_1$ only maintains consistency in scenarios (i) and (ii), while exhibiting notable bias when either weighting model is incorrectly specified in scenarios (iii) and (iv). The 95\% coverage probabilities of the three semiparametric estimators $\widehat{\tau}_1$, $\widehat{\tau}_2$, and $\widehat{\tau}_3$ align with the nominal level solely when their underlying models are accurately specified, whereas the 95\% coverage probabilities of $\widehat{\tau}_{\MR}$ closely approximate the nominal level in all four scenarios.
	These results confirm our previous theoretical findings and demonstrate the advantages of the proposed triply robust estimator. We provide additional simulation studies for a binary outcome in the supplementary material, which corroborate    the simulation results presented here.

	\subsection{Empirical example}  
	\label{sec: application}
	In this section, we investigate the causal effect of physical activity on healthcare expenditures among adults in the United States. We analyze the primary data from the 2018 Medical Expenditure Panel Survey (MEPS). The primary dataset originates from the 2018 full-year consolidated file of MEPS, mainly focusing on demographics, health status, and health insurance. Additionally, we utilize an auxiliary data also obtained from MEPS, specifically drawn  from the 2018 job file. This auxiliary dataset  is designed to collect supplementary job-related information.

	The primary dataset consists of the following variables: the treatment variable $Z$, which indicates whether an individual engages in at least half an hour of moderate to vigorous physical activity at least five times a week; the outcome variable $Y$, which represents healthcare expenditures; and the baseline covariates $X$, including factors such as age, sex, race, marital status, BMI, education, family poverty level, smoking status, cancer diagnosis, and health insurance coverage.  The primary dataset consists of 18,774 subjects, with 9,174 subjects having $Z=1$ and 9,600 subjects having $Z=0$. The auxiliary dataset consists of  6,118 subjects, and the resulting linked  dataset includes 5,226 subjects. This linked  dataset incorporates additional significant confounding variables, including working hours and income, denoted by   $V$.

	Table S2 in the supplementary material presents highly similar summary statistics of the demographic characteristics between the two datasets. Furthermore, under the MAR assumption, we conducted an empirical test to assess Assumption \ref{ass: MAR-V}. This test involved using logistic regression to examine the selection mechanism  based on the fully observed variables $(Z, X, Y)$. The estimated coefficient of $Y$ is statistically significant but very small, approximately less than 10$^{-5}$, while the estimated coefficient of $Z$ is not significant. These results  suggest that there is no significant evidence of violating Assumption \ref{ass: MAR-V} when the MAR assumption holds.   
	
	We employ a linear model for the outcome regression model, a normal distribution  for the imputation model, and logistic models for the selection mechanism and propensity score. The asymptotic variance is estimated using the nonparametric bootstrap method. Table \ref{tab: application-res} presents the estimated point estimates, empirical standard errors, and 95\% confidence intervals for the four estimators. The results from  four estimators are statistically significant and demonstrate a consistent pattern, suggesting that regular physical activity may lead to reduced healthcare costs. These findings align with previous research  \citep{kokkinos2011physical,booth2012lack}.
	
	\begin{table*}[t]
		\caption{Estimates of the causal effect of physical activities on total healthcare expenditures (scaled by 1000 US dollars). 
		}
		\label{tab: application-res}
		\centering
		\resizebox{0.9\textwidth}{!}{
			\tabcolsep=0pt
			\begin{tabular*}{\textwidth}{@{\extracolsep{\fill}}cccc@{\extracolsep{\fill}}}
				\toprule
				Estimators      &    Point estimate    & Standard error   & 95$\%$ Confidence interval  \\ 
				\hline\addlinespace[1mm]
    		$\widehat\tau_1~~$   &     -1.00        & 0.51        & (-1.99, -0.01)       \\
				$\widehat\tau_2~~$   &    -1.06       & 0.48           & (-2.00, -0.11)         \\
				$\widehat\tau_3~~$   &   -1.20        & 0.55        & (-2.28, -0.12)         \\
				$\widehat\tau_\MR$ &    -1.35       & 0.56         & (-2.45,  -0.24)        \\
				\bottomrule
			\end{tabular*}
		}
	\end{table*}
	
	\section{Disscussion} 
	\label{sec:concludes}
	It is crucial to   consider potential selection bias when estimating  causal effects from the linked data. Neglecting this bias could lead to biased estimates and incorrect conclusions.  Our study has introduced three nonparametric strategies for identifying ATE. Moreover, we have presented a triply robust and locally efficient estimator that is easy to implement.  We have also investigated the asymptotic properties of our proposed estimator, showcasing its consistency and asymptotic normality even when employing various flexible machine learning techniques for nuisance model estimation. Additionally, we have discussed potential extensions of the proposed method, including optimal sampling designs and the estimation of nonlinear causal estimands.

	There are several possible extensions for future work.    Firstly, while this paper focuses on constructing a triply robust estimator from two data sources, it would be of interest to develop theories and methodologies for the scenario with multiple data sources. 
	Secondly, the problem of optimal sample size and sample rate allocation could be further explored under  Assumption \ref{ass: MAR-V} or Assumption \ref{assumption: MAR2}.   
	The study of these issues is beyond the scope of this paper and we leave them as future research topics.

		\begin{landscape}
		\begin{table}
			\caption{Bias, standard deviation (SD), and 95\% coverage probability (CP) of the proposed estimators with a continuous outcome under different sample sizes. The bias, SD and CP are multiplied by $10^2$. 
   }
			\label{tab: MC-prop}
			\centering
			\begin{threeparttable}    
				\renewcommand{\arraystretch}{0.75}
				\resizebox{0.835\columnwidth}{!}{%
				\begin{tabular}{cccccccccccccccccccc}
    \toprule
    
    \specialrule{0em}{3pt}{3pt}
    
    \multicolumn{3}{c}{Model} & 
    \multicolumn{1}{c}{} &
    \multicolumn{3}{c}{$\widehat\tau_{1}$} & 
    \multicolumn{1}{c}{} &
    \multicolumn{3}{c}{$\widehat\tau_{2}$} & 
    \multicolumn{1}{c}{} &
    \multicolumn{3}{c}{$\widehat\tau_{3}$}& 
    \multicolumn{1}{c}{} &
    \multicolumn{3}{c}{$\widehat\tau_{\MR}$}  \\
    \specialrule{0em}{3pt}{3pt}

    \cmidrule(lr){1-3}	\cmidrule(lr){5-7} 	\cmidrule(lr){9-11}	\cmidrule(lr){13-15}	\cmidrule(lr){17-19}
    
    \multirow{2}{*}{$\mathcal{M}_1$} &	\multirow{2}{*}{$\mathcal{M}_2$} & 	\multirow{2}{*}{$\mathcal{M}_3$} &\multicolumn{1}{c}{}
    &\multicolumn{3}{c}{Sample Size}&\multicolumn{1}{c}{}&\multicolumn{3}{c}{Sample Size}&\multicolumn{1}{c}{}&\multicolumn{3}{c}{Sample Size}&\multicolumn{1}{c}{}&\multicolumn{3}{c}{Sample Size} \\
    \cmidrule(lr){5-7} 	\cmidrule(lr){9-11}	\cmidrule(lr){13-15}	\cmidrule(lr){17-19}
    & &	 & & 1000  & 5000 & 10000 	&& 1000  & 5000 & 10000	&& 1000  & 5000 & 10000 && 1000  & 5000 & 10000 \\
    \specialrule{0em}{0pt}{3pt}
    \midrule
    
    &Bias &  &  & & & & & & & & & & &&&&&\\
    $\checkmark $& $\checkmark $& $\checkmark $ & & -1 & 1 & 0 & & 0 & 1 & 0& &0 & 1 & 0& & 0 & 1 & 0\\
    $\checkmark $ & $\times$& $\times$ & & -1 & 1 & 0& & 157 & 158 & 158& &162 & 164 & 163&& -2 & 2 & 0\\
    $\times$ & $\checkmark $& $\times$ && 117 & 119 & 118&& 0 & 1 & 0& &42 & 43 & 43&& 0 & 1 & 0\\
    $\times$ & $\times$& $\checkmark $&& 163 & 164 & 163&& 40 & 41 & 40&& 0 & 1 & 0&& 0 & 1 & 0\\
    $\times$ & $\times$& $\times$&& 163 & 164 & 163&& 162 & 164 & 163 && 162 & 164 & 163&& 163 & 165 & 164 \\
    \addlinespace[1mm]
    
    & { SD } &  &  & & & & & & & & & & &&&&&\\
    $\checkmark $& $\checkmark $& $\checkmark $ && 18 & 7 & 5&& 12 & 6 & 4&& 12 & 6 & 4&& 12 & 6 & 4  \\
    $\checkmark $ & $\times$& $\times$  && 18 & 7 & 5 && 19 & 8 & 6&& 19 & 8 & 6&& 31 & 12 & 8   \\
    $\times$ & $\checkmark $& $\times$ && 19 & 7 & 5 && 12 & 6 & 4&& 14 & 9 & 7&& 12 & 6 & 4     \\
    $\times$ & $\times$& $\checkmark $ && 19 & 8 & 6&& 13 & 6 & 4 && 12 & 6 & 4&& 12 & 6 & 4\\
    $\times$ & $\times$& $\times$ && 19 & 7 & 6&& 19 & 8 & 6&& 19 & 8 & 6 && 19 & 8 & 6\\
    \addlinespace[1mm]
    
    &    CP  &  &  & & & & & & & & & & &&&&&\\
    $\checkmark $& $\checkmark $& $\checkmark $ && 95 & 95 & 96&& 97 & 95 & 95&& 97 & 95 & 94&& 97 & 96 & 94\\
    $\checkmark $ & $\times$& $\times$ && 95 & 95 & 96&& 0 & 0 & 0&& 0 & 0 & 0&& 93 & 93 & 94\\
    $\times$ & $\checkmark $& $\times$ && 0 & 0 & 0&& 97 & 95 & 95&& 17 & 5 & 0&& 97 & 97 & 96\\
    $\times$ & $\times$& $\checkmark $ && 0 & 0 & 0&& 19 & 0 & 0 && 97 & 95 & 94&& 96 & 96 & 95\\
    $\times$ & $\times$& $\times$ && 0 & 0 & 0&& 0 & 0& 0&& 0 & 0 & 0&& 0 & 0 & 0\\

    \bottomrule
\end{tabular}
} 
			\end{threeparttable}       
		\end{table}
	\end{landscape}
	
	
	
	


	\bibliographystyle{apalike}

	\bibliography{reference}

\begin{thebibliography}{}

\bibitem[Bang and Robins, 2005]{bang2005doubly}
Bang, H. and Robins, J.~M. (2005).
\newblock {Doubly Robust Estimation in Missing Data and Causal Inference
  Models}.
\newblock {\em Biometrics}, 61(4):962--973.

\bibitem[Bickel et~al., 1993]{bickel1993efficient}
Bickel, P.~J., Klaassen, C., Ritov, Y., and Wellner, J.~A. (1993).
\newblock {\em Efficient and Adaptive Inference in Semiparametric Models}.
\newblock Johns Hopkins University Press, Baltimore, MD.

\bibitem[Booth et~al., 2012]{booth2012lack}
Booth, F.~W., Roberts, C.~K., and Laye, M.~J. (2012).
\newblock {Lack of Exercise is a Major Cause of Chronic Diseases}.
\newblock {\em Comprehensive Physiology}, 2(2):1143.

\bibitem[Chatterjee et~al., 2003]{chatterjee2003pseudoscore}
Chatterjee, N., Chen, Y.-H., and Breslow, N.~E. (2003).
\newblock {A Pseudoscore Estimator for Regression Problems with Two-Phase
  Sampling}.
\newblock {\em Journal of the American Statistical Association},
  98(461):158--168.

\bibitem[Chen, 2007]{yun2007semiparametric}
Chen, H.~Y. (2007).
\newblock {A Semiparametric Odds Ratio Model for Measuring Association}.
\newblock {\em Biometrics}, 63(2):413--421.

\bibitem[Chernozhukov et~al., 2018]{chernozhukov2018double}
Chernozhukov, V., Chetverikov, D., Demirer, M., Duflo, E., Hansen, C., Newey,
  W., and Robins, J. (2018).
\newblock {Double/Debiased Machine Learning for Treatment and Structural
  Parameters}.
\newblock {\em The Econometrics Journal}, 21(1):C1--C68.

\bibitem[Cochran, 2007]{cochran2007sampling}
Cochran, W.~G. (2007).
\newblock {\em Sampling Techniques}.
\newblock Wiley, New York, 3rd edition.

\bibitem[Evans et~al., 2020]{evans2020coherent}
Evans, K., Fulcher, I., and Tchetgen, E. J.~T. (2020).
\newblock {A Coherent Likelihood Parametrization for Doubly Robust Estimation
  of a Causal Effect with Missing Confounders}.
\newblock {\em arXiv:2007.10393}.

\bibitem[Farrell et~al., 2021]{farrell2021deep}
Farrell, M.~H., Liang, T., and Misra, S. (2021).
\newblock {Deep Neural Networks for Estimation and Inference}.
\newblock {\em Econometrica}, 89(1):181--213.

\bibitem[H{\'a}jek, 1971]{hajek1971comment}
H{\'a}jek, J. (1971).
\newblock {Comment on a Paper by D. Basu}.
\newblock In Godambe, V.~P. and Sprott, D.~A., editors, {\em Foundations of
  Statistical Inference}, page 236. Holt, Rinehart and Winston, Toronto.

\bibitem[Horvitz and Thompson, 1952]{horvitz1952generalization}
Horvitz, D.~G. and Thompson, D.~J. (1952).
\newblock {A Generalization of Sampling Without Replacement from a Finite
  Universe}.
\newblock {\em Journal of the American Statistical Association},
  47(260):663--685.

\bibitem[Jiang et~al., 2022]{Jiang2022multiply}
Jiang, Z., Yang, S., and Ding, P. (2022).
\newblock {Multiply Robust Estimation of Causal Effects Under Principal
  Ignorability}.
\newblock {\em Journal of the Royal Statistical Society: Series B (Statistical
  Methodology)}, 84(4):1423--1445.

\bibitem[Kang and Schafer, 2007]{kang2007demystifying}
Kang, J.~D. and Schafer, J.~L. (2007).
\newblock {Demystifying Double Robustness: a Comparison of Alternative
  Strategies for Estimating a Population Mean from Incomplete Data}.
\newblock {\em Statistical Science}, 22(4):523--539.

\bibitem[Kennedy, 2016]{kennedy2016semiparametric}
Kennedy, E.~H. (2016).
\newblock {Semiparametric Theory and Empirical Processes in Causal Inference}.
\newblock In {\em Statistical Causal Inferences and Their Applications in
  Public Health Research}, pages 141--167. Springer.

\bibitem[Kennedy, 2022]{kennedy2023semiparametric}
Kennedy, E.~H. (2022).
\newblock {Semiparametric Doubly Robust Targeted Double Machine Learning: a
  Review}.
\newblock {\em arXiv:2203.06469}.

\bibitem[Kokkinos et~al., 2011]{kokkinos2011physical}
Kokkinos, P., Sheriff, H., and Kheirbek, R. (2011).
\newblock {Physical Inactivity and Mortality Risk}.
\newblock {\em Cardiology Research and Practice}, 2011:924945.

\bibitem[Lin and Chen, 2014]{lin2014adjustment}
Lin, H.-W. and Chen, Y.-H. (2014).
\newblock {Adjustment for Missing Confounders in Studies Based on Observational
  Databases: 2-Stage Calibration Combining Propensity Scores from Primary and
  Validation Data}.
\newblock {\em American Journal of Epidemiology}, 180(3):308--317.

\bibitem[Mayer et~al., 2020]{mayer2020doubly}
Mayer, I., Sverdrup, E., Gauss, T., Moyer, J.-D., Wager, S., and Josse, J.
  (2020).
\newblock {Doubly Robust Treatment Effect Estimation with Missing Attributes}.
\newblock {\em The Annals of Applied Statistics}, 14(3):1409--1431.

\bibitem[McCandless et~al., 2012]{mccandless2012adjustment}
McCandless, L.~C., Richardson, S., and Best, N. (2012).
\newblock {Adjustment for Missing Confounders Using External Validation Data
  and Propensity Scores}.
\newblock {\em Journal of the American Statistical Association},
  107(497):40--51.

\bibitem[Mitra and Reiter, 2011]{Mitra2011EstimatingPS}
Mitra, R. and Reiter, J.~P. (2011).
\newblock {Estimating Propensity Scores with Missing Covariate Data Using
  General Location Mixture Models}.
\newblock {\em Statistics in Medicine}, 30(6):627--641.

\bibitem[Newey and McFadden, 1994]{newey1994large}
Newey, W.~K. and McFadden, D. (1994).
\newblock {Large Sample Estimation and Hypothesis Testing}.
\newblock {\em Handbook of Econometrics}, 4:2111--2245.

\bibitem[Newey and Robins, 2018]{newey2018cross}
Newey, W.~K. and Robins, J.~R. (2018).
\newblock {Cross-Fitting and Fast Remainder Rates for Semiparametric
  Estimation}.
\newblock {\em arXiv:1801.09138}.

\bibitem[Robins et~al., 1994]{robins1994estimation}
Robins, J.~M., Rotnitzky, A., and Zhao, L.~P. (1994).
\newblock {Estimation of Regression Coefficients when Some Regressors are not
  Always Observed}.
\newblock {\em Journal of the American Statistical Association},
  89(427):846--866.

\bibitem[Rosenbaum and Rubin, 1983]{rosenbaum1983central}
Rosenbaum, P.~R. and Rubin, D.~B. (1983).
\newblock {The Central Role of the Propensity Score in Observational Studies
  for Causal Effects}.
\newblock {\em Biometrika}, 70(1):41--55.

\bibitem[Rubin, 1974]{rubin1974estimating}
Rubin, D.~B. (1974).
\newblock {Estimating Causal Effects of Treatments in Randomized and
  Nonrandomized Studies}.
\newblock {\em Journal of Educational Psychology}, 66(5):688.

\bibitem[Seaman and White, 2014]{seaman2014inverse}
Seaman, S. and White, I. (2014).
\newblock {Inverse Probability Weighting with Missing Predictors of Treatment
  Assignment or Missingness}.
\newblock {\em Communications in Statistics-Theory and Methods},
  43(16):3499--3515.

\bibitem[St{\"u}rmer et~al., 2005]{sturmer2005adjusting}
St{\"u}rmer, T., Schneeweiss, S., Avorn, J., and Glynn, R.~J. (2005).
\newblock {Adjusting Effect Estimates for Unmeasured Confounding with
  Validation Data Using Propensity Score Calibration}.
\newblock {\em American Journal of Epidemiology}, 162(3):279--289.

\bibitem[Sun et~al., 2022]{sun2022use}
Sun, J.~W., Wang, R., Li, D., and Toh, S. (2022).
\newblock {Use of Linked Databases for Improved Confounding Control:
  Considerations for Potential Selection Bias}.
\newblock {\em American Journal of Epidemiology}, 191(4):711--723.

\bibitem[Van~der Vaart, 2000]{van2000asymptotic}
Van~der Vaart, A.~W. (2000).
\newblock {\em {Asymptotic Statistics}}, volume~3.
\newblock Cambridge University Press.

\bibitem[Wang and Tchetgen~Tchetgen, 2018]{wang2018bounded}
Wang, L. and Tchetgen~Tchetgen, E. (2018).
\newblock {Bounded, Efficient and Multiply Robust Estimation of Average
  Treatment Effects Using Instrumental Variables}.
\newblock {\em Journal of the Royal Statistical Society: Series B (Statistical
  Methodology)}, 80(3):531--550.

\bibitem[Wang et~al., 2009]{wang2009causal}
Wang, W., Scharfstein, D., Tan, Z., and MacKenzie, E.~J. (2009).
\newblock {Causal Inference in Outcome-Dependent Two-Phase Sampling Designs}.
\newblock {\em Journal of the Royal Statistical Society Series B: Statistical
  Methodology}, 71(5):947--969.

\bibitem[Williamson et~al., 2012]{williamson2012doubly}
Williamson, E.~J., Forbes, A., and Wolfe, R. (2012).
\newblock {Doubly Robust Estimators of Causal Exposure Effects with Missing
  Data in the Outcome, Exposure or a Confounder}.
\newblock {\em Statistics in Medicine}, 31(30):4382--4400.

\bibitem[Yang and Ding, 2020]{yang2019combining}
Yang, S. and Ding, P. (2020).
\newblock {Combining Multiple Observational Data Sources to Estimate Causal
  Effects}.
\newblock {\em Journal of the American Statistical Association},
  115(531):1540--1554.
\newblock PMID: 33088006.

\end{thebibliography}

\renewcommand{\theproposition}{S\arabic{proposition}}
\renewcommand{\thetheorem}{S\arabic{theorem}}
\renewcommand{\theassumption}{S\arabic{assumption}}
\renewcommand{\thesection}{S\arabic{section}}
\renewcommand{\theequation}{S\arabic{equation}}
\renewcommand{\thelemma}{S\arabic{lemma}} 
\renewcommand{\thetable}{S\arabic{table}}

{\centering \section*{Supplementray Material}}
	  In the Supplementary Material, we provide proofs of theorems and propositions in
the main paper. We also provide additional details for the simulations and data
application.

\setcounter{section}{0}
\section{The proof of Theorem \ref{THEOREM:IDENTIFIABLE}}
\label{sec:proof-tr-iden}
\begin{proof}
We proceed as follows:
    \begin{itemize}
        \item[(a)] Based on  selection probability and propensity score,  	\begin{align*} 
&\E(Y_1-Y_0)\\&=\E\left\{\E\left(Y_1-Y_0\mid X,V\right)\right\}\\&=\E\left\{\E\left(Y_1\mid X,V\right)\right\}-\E\left\{\E\left(Y_0\mid X,V\right)\right\}\\&=\E\left\{\E\left(Y_1\mid Z,X,V\right)\right\}-\E\left\{\E\left(Y_0\mid Z,X,V\right)\right\}\\&=\E\left\{\E\left(Y\mid Z=1,X,V\right)\right\}-\E\left\{\E\left(Y\mid Z=0,X,V\right)\right\}\\&=\E\left[\E\left\{\left.\frac{ZY}{\pi(X,V)}\right|X,V\right\}\right]-\E\left[\E\left\{\left.\frac{\left(1-Z\right)Y}{1-\pi(X,V)}\right|X,V\right\}\right]\\&=\E\left[\E\left\{\left.\frac{ZY}{\pi(X,V)}\right|X,V,R=1\right\}\right]-\E\left[\E\left\{\left.\frac{\left(1-Z\right)Y}{1-\pi(X,V)}\right|X,V,R=1\right\}\right]\\&=\E\left[\E\left\{\left.\frac R{\rho(X)}\frac{ZY}{\pi(X,V)}-\frac R{\rho(X)}\frac{\left(1-Z\right)Y}{1-\pi(X,V)}\right|X,V\right\}\right]\\&=\E\left\{\frac R{\rho(X)}\frac{ZY}{\pi(X,V)}-\frac R{\rho(X)}\frac{\left(1-Z\right)Y}{1-\pi(X,V)}\right\}.
	\end{align*}
 \item[(b)] Based on selection probability and outcome mean, 
 \begin{align*}
 \E&\left(Y_1-Y_0\right)\\&=\E\left\{\E\left(Y_1-Y_0\mid X,V\right)\right\}\\&=\E\left\{\E\left(Y_1\mid X,V\right)\right\}-\E\left\{\E\left(Y_0\mid X,V\right)\right\}\\&=\E\left\{\E\left(Y_1\mid Z,X,V\right)\right\}-\E\left\{\E\left(Y_0\mid Z,X,V\right)\right\}\\&=\E\left\{\E\left(Y\mid Z=1,X,V\right)\right\}-\E\left\{\E\left(Y\mid Z=0,X,V\right)\right\}\\&=\E\left\{\E\left(Y\mid Z=1,X,V,R=1\right)\right\}-\E\left\{\E\left(Y\mid Z=0,X,V,R=1\right)\right\}\\&=\E\left[\E\left\{\left.\frac{RY}{\pr(R\mid Z=1,X,V)}\right|Z=1,X,V\right\}\right]\\&~~~~~~~-\E\left[\E\left\{\left.\frac{RY}{\pr(R\mid Z=0,X,V)}\right|Z=0,X,V\right\}\right]\\&=\E\left\{\E\left\{\left.\frac{RY}{\rho(X)}\right|Z=1,X,V\right\}\right\}-\E\left\{\E\left\{\left.\frac{RY}{\rho(X)}\right|Z=0,X,V\right\}\right\} {.}
	\end{align*}
 \item[(c)] Based on imputation model and outcome mean, we have
 \begin{align*}
\E(Y_1-Y_0\mid X )  &=
\E\{\mu(1,X,V)- \mu(0,X,V )\mid X \}= \delta(1,X)-\delta(0,X) {,}
\end{align*}  
where $\mu(Z,X,V)=\E(Y\mid Z,X,V,R=1)$ is identifiable from the  observed data. Therefore, we have
 \begin{align*}
\E (Y_1-Y_0 ) =\E\left\{\delta(1,X)-\delta(0,X)\right\} .
	\end{align*}
    \end{itemize}
\end{proof} 
 \section{The proof of Theorem \ref{THM: EIF}}
 \begin{proof}
  
 The observed likelihood and full-data likelihood  can be factorized as:   
 \begin{align*}
     f_{\mathrm{obs}}&(X,V,Z,Y,R) =f(X,V,Z,Y,R=1  )^R f(X, Z,Y,R=0 )^{(1-R)},\\
     f_{\mathrm{full}}&(X,V,Z,Y,R )=f(X,V,Z,Y,R=1 )^R f(X,V, Z,Y,R=0 )^{(1-R)} {,}
 \end{align*} 
 To derive the EIFs, we consider a one-dimensional parametric submodel, $f  (X,V,Z,Y,R;t )$, which contains the true model $f(X,V,Z,Y,R )$ at $t = 0$, that is, $f   (X,V,Z,Y,R;t )\vert_{t=0} = f(X,V,Z,Y,R )$. We use $t$ in the subscript to denote the quantity with respect to the submodel. 
Then, the observed and full-data score functions of the submodel are: 
  \begin{align*}
      \begin{small}
          \begin{aligned}
S_{\mathrm{obs}}(X,V,Z,Y,R;t)&=\partial\log f_{\mathrm{obs}}(X,V,Z,Y,R;t)/\partial t\\&=R{\partial}\log f(X,V,Z,Y,R=1;t){/\partial t} +(1-R){\partial\log}f(X,Z,Y,R=0;t){/\partial t}.\\S_{\mathrm{full}}(X,V,Z,Y,R;t)&=\partial\log f_{\mathrm{full}}(X,V,Z,Y,R;t)/\partial t\\&=R{\partial}\log f(X,V,Z,Y,R=1;t){/\partial t} +(1-R){\partial}\log f(X,V,Z,Y,R=0;t){/\partial t}.
 \end{aligned} 
      \end{small}
  \end{align*}
 Hence, we have that $$RS_{\mathrm{obs}} (X,V,Z,Y,R; {t})=RS_{\mathrm{full}}(X,V,Z,Y,R ; {t}).$$
        We will use the semiparametric theory in \citet{bickel1993efficient} to derive the EIFs.  The EIF for $\tau  $, denoted by $\varphi (O)$, must satisfy
\begin{align*}
\left.\frac{\partial}{\partial {t}} \tau_{{ t}}\right|_{{t}=0}=\left.\E\{\varphi (O) \cdot S_\mathrm{obs
}
(Y, Z, X,V,R ; {t})\}\right|_{{t}=0}{ ,}
\end{align*} 
  Let $m_{{t}}(X,V)=\mu_{t}(1,X,V)-\mu_{t}(0,X,V)$. From Theorem \ref{THEOREM:IDENTIFIABLE}(b) and the chain rule, we have  that
\begin{align} 
 {\left.\frac\partial{\partial{t}}\tau_{t}\right|}_{{t}=0} &\notag={\left.\frac\partial{\partial{t}}\E_{t}\left\{\frac{Rm_{t}(X,V)}{{\rho}_{t}(X)}\right\}\right|}_{{t}=0}\\\;\;\;\;\;\;&\notag={\left.\frac\partial{\partial{t}}\E_{t}\left\{\frac{Rm(X,V)}{{\rho}(X)}\right\}\right|}_{{t}=0}+{\left.\E\left\{\frac{{{R}}\frac\partial{\partial{t}}m_{t}(X,V){\rho}_{t}(X)-{{R}}{m_{t}(X,V)}\frac\partial{\partial{t}}{\rho}_{t}(X)}{\left\{{\rho}_{t}(X)\right\}^2}\right\}\right|}_{{t}=0}\\&= {\left.\frac\partial{\partial{t}}\E_{t}\left\{\frac{Rm(X,V)}{{\rho}(X)}\right\}\right|}_{{t}=0}\label{eq: S1-term} \\&~~~+ \E\left\{\frac R{{\rho}(X)}{\left.\frac\partial{\partial{t}}m_{t}(X,V)\right|}_{{t}=0}\right\}\label{eq: S2-term}  \\&~~~- \E\left\{\frac{Rm(X,V)}{\left\{{\rho}(X)\right\}^2}{\left.\frac\partial{\partial{t}}{\rho}_{t}(X)\right|}_{{t}=0}\right\} \label{eq: S3-term} {{.}}
\end{align} 
\subsection{Calculation of \eqref{eq: S1-term}}
\label{ssec:cal-S1}
\begin{align*}
    {\left.\frac\partial{\partial{t}}\E_{t}\left\{\frac{Rm(X,V)}{{\rho}(X)}\right\}\right|}_{{t}=0}& =\E\left\{\frac{Rm(X,V)}{{\rho}(X)}S_\full(X,V,R)\right\}\\& =\E\left\{\frac{Rm(X,V)}{{\rho}(X)}S_\full(X,V,R) \right\}+\E\left\{\frac{Rm(X,V)}{{\rho}(X)}S_\full(Z,Y\mid  X,V,R)\right\} \\& =\E\left\{\frac{Rm(X,V)}{{\rho}(X)}S_\full(X,V,R,Z,Y))\right\}\\& =\E\left\{\frac{Rm(X,V)}{{\rho}(X)}S_\obs(X,V,R,Z,Y))\right\}. 
\end{align*} 
\subsection{Calculation of \eqref{eq: S2-term}}\label{ssec:cal-S2}
\begin{align*}
&{\left.\frac\partial{\partial{t}}\E_{t}(Y\mid Z=z,X,V,R)\right|}_{{t}=0}\\&~~~ =\E\left\{YS_\full(Y\mid Z=z,X,V,R)\mid Z=z,X,V,R\right\}\\&~~~ =\E\left\{YS_\full(Y\mid Z=z,X,V,R)\mid Z=z,X,V,R\right\}\\&\;\;\;\;\;\;\;\;\;\;\;-\E(Y\mid Z=z,X,V,R)\underbrace{\E\left\{S_\full(Y\mid Z=z,X,V,R)\mid Z=z,X,V,R\right\}}_0\\&~~~ =\E\left\{YS_\full(Y\mid Z=z,X,V,R)\mid Z=z,X,V,R\right\}\\&\;\;\;\;\;\;\;\;\;\;\;-\E\left\{\E(Y\mid Z=z,X,V,R)S_\full(Y\mid Z=z,X,V,R)\mid Z=z,X,V,R\right\}\\&~~~ =\E\left[\left\{Y-\E(Y\mid Z=z,X,V,R)\right\}S_\full(Y\mid Z=z,X,V,R)\mid Z=z,X,V,R\right]\\&\;\;\;\;\;\;\;\;\;+\E\left[\left\{Y-\E(Y\mid Z=z,X,V,R)\right\}S_\full(Z=z,X,V,R)\mid Z=z,X,V,R\right] \\&~~~ =\E\left[\left\{Y-\E(Y\mid Z=z,X,V,R)\right\}S_\full(Z=z,X,V,R,Y)\mid Z=z,X,V,R\right] {.}
\end{align*}
Hence,
\begin{align*}&{\left.\frac\partial{\partial t}m_t^{}(X)\right|}_{\;t=0}\\&={\left.\frac\partial{\partial t}\E_t(Y\mid Z=1,X,V,R)\right|}_{t=0}-{\left.\frac\partial{\partial t}\E_t(Y\mid Z=0,X,V,R)\right|}_{t=0}\\&=\E\left[\left\{Y-\mu(1,X,V)\right\}S_\full(Y,R,Z=1,X,V)\mid Z=1,X,V,R\right]\\&\;\;\;\;\;-\E\left[\left\{Y-\mu(0,X,V)\right\}S_\full(Y,R,Z=0,X,V)\mid Z=0,X,V,R\right]\\&=\E\left[\left.\frac Z{\pi(X,V)}\left\{Y-\mu(1,X,V)\right\}S_\full(Z,X,V,Y,R)\right|X,V,R\right]\\&\;\;\;\;-\E\left[\left.\frac{1-Z}{1-\pi(X,V)}\left\{Y-\mu(0,X,V)\right\}S_\full(Z,X,V,Y,R)\right|X,V,R\right] {.}
\end{align*} 
We then have,
\begin{align*}&\E\left\{\frac R{{{\rho}}(X)}{\left.\frac\partial{\partial{t}}m_{t}(X,V)\right|}_{{t}=0}\right\}\\&=\E\left(\frac R{\rho(X)}\E\left[\left.\frac Z{\pi(X,V)}\left\{Y-\mu(1,X,V)\right\}S_\full(Z,X,V,Y,R)\right|X,V,R\right]\right)\\&\;\;\;-\E\left(\frac R{\rho(X)}\E\left[\left.\frac{1-Z}{1-\pi(X,V)}\left\{Y-\mu(0,X,V)\right\}S_\full(Z,X,V,Y,R)\right|X,V,R\right]\right)\\&=\E\left(\E\left[\frac R{\rho(X)}\left.\frac Z{\pi(X,V)}\left\{Y-\mu(1,X,V)\right\}S_\full(Z,X,V,Y,R)\right|X,V,R\right]\right)\\&\;\;\;\;\;-\E\left(\E\left[\left.\frac R{\rho(X)}\frac{1-Z}{1-\pi(X,V)}\left\{Y-\mu(0,X,V)\right\}S_\full(Z,X,V,Y,R)\right|X,V,R\right]\right)\\&=\E\left[\frac R{\rho(X)}\frac Z{\pi(X,V)}\left\{Y-\mu(1,X,V)\right\}S_\full(Z,X,V,Y,R)\right]\\&\;\;\;\;\;-\E\left[\frac R{\rho(X)}\frac{1-Z}{1-\pi(X,V)}\left\{Y-\mu(0,X,V)\right\}S_\full(Z,X,V,Y,R)\right]\\&=\E\left[\frac R{\rho(X)}\frac Z{\pi(X,V)}\left\{Y-\mu(1,X,V)\right\}S_\obs(Z,X,V,Y,R)\right]\\&\;\;\;\;\;-\E\left[\frac R{\rho(X)}\frac{1-Z}{1-\pi(X,V)}\left\{Y-\mu(0,X,V)\right\}S_\obs(Z,X,V,Y,R)\right] {.}
\end{align*}

\subsection{Calculation of \eqref{eq: S3-term}}\label{ssec:cal-S3}
\begin{align*}
{\left.\frac\partial{\partial{t}}{\rho}_{t}(X)\right|}_{\;{t}=0}&={\left.\frac\partial{\partial{t}}\mathrm{pr}(R=1\mid X ;{t})\right|}_{{t}=0}\\&={\left.\frac{\frac{\partial}{\partial{t}}\mathrm{pr}(R=1\mid X ;{t})}{\mathrm{pr}(R=1\mid X ;{t})}\mathrm{pr}(R=1\mid X ;{t})\right|}_{{t}=0}\\&=S_\full(R=1\mid X )\mathrm{pr}(R=1\mid X) \\&=S_\full(R=1\mid X ){\rho}( X) {,}~~~~~~~~~~~~~~~~~~~~~~~~~~~~~~~~~~~~~~~~~~~~~~~~~~~~~~~~~~~~~~~~~~~~~~~~~~~~~~~~~~~~~~~~~~~~~~~~~~~~~~~~~~~~~~~~~~~~~~~~~~~~~~~~~~~~
\end{align*}
\begin{align*} {\E}&\left\{\frac{Rm(X,V)}{\left\{\rho(X)\right\}^2}{\left.\frac\partial{\partial t}\rho_t(X)\right|}_{t=0}\right\}\\&={\E}\left\{\frac{Rm(X,V)S_{\full}(R=1\mid X)\rho(X)}{\left\{\rho(X)\right\}^2}\right\}\\&={\E}\left\{\frac{Rm(X,V)S_{\full}(R=1\mid X)}{\rho(X)}\right\}\\&={\E}\left\{\frac R{\rho(X)}{\E}\left\{m(X,V)\mid X\right\}S_{\full}(R\mid X)\right\}\\&={\E}\left[\frac{R{\E}\left\{m(X,V)\mid X\right\}S_{\full}(R\mid X)}{\rho(X)}\right]-\underbrace{{\E}\left[{\E}\left\{m(X,V)\mid X\right\}S_{\full}(R\mid X)\right]}_0\\&={\E}\left\{\frac{\left\{R-\rho(X)\right\}{\E}\left\{m(X,V)\mid X\right\}S_{\full}(R\mid X)}{\rho(X)}\right\}\\&={\E}\left\{\frac{\left\{R-\rho(X)\right\}{\E}\left\{m(X,V)\mid X\right\}\left[S_{\full}(R\mid X)+S_{\full}(X)+S_{\full}(V\mid X)\right]}{\rho(X)}\right\}\\&\;\;\;\;\;\;+{\E}\left\{\frac{\left\{R-\rho(X)\right\}{\E}\left\{m(X,V)\mid X\right\}\left[S_{\full}(Z\mid X,V)+S_{\full}(Y\mid Z,X,V)\right]}{\rho(X)}\right\}\\&={\E}\left[\frac{\left\{R-\rho(X)\right\}{\E}\left\{m(X,V)\mid X\right\}S_{\full}(Z,X,V,Y,R)}{\rho(X)}\right]\\&={\E}\left[\frac{\left\{R-\rho(X)\right\}{\E}\left\{m(X,V)\mid X\right\}S_{\obs}(Z,X,V,Y,R)}{\rho(X)}\right],
\end{align*}  
where the final equality holds due to the fact
\begin{align*}
    \E&\left[\frac{\left\{R-\rho(X)\right\}\E\left\{m(X,V)\mid X\right\}S_{\full}(Z,X,V,Y,R)}{\rho(X)}\right]\\&=\E\left[\frac{R\E\left\{m(X,V)\mid X\right\}S_{\obs}(Z,X,V,Y,R)}{\rho(X)}\right]-\E\left[\frac{  \rho(X) \E\left\{m(X,V)\mid X\right\}S_{\full}(Z,X,V,Y,R)}{\rho(X)}\right]\\&=\E\left[\frac{R\E\left\{m(X,V)\mid X\right\}S_{\obs}(Z,X,V,Y,R)}{\rho(X)}\right]-\E\left[\frac{  \rho(X) \E\left\{m(X,V)\mid X\right\}S_{\full}(X)}{\rho(X)}\right]\\&=\E\left[\frac{R\E\left\{m(X,V)\mid X\right\}S_{\obs}(Z,X,V,Y,R)}{\rho(X)}\right]-\E\left[\frac{  \rho(X) \E\left\{m(X,V)\mid X\right\}S_{\obs}(X)}{\rho(X)}\right] \\&=\E\left[\frac{\left\{R-\rho(X)\right\}\E\left\{m(X,V)\mid X\right\}S_{\obs}(Z,X,V,Y,R)}{\rho(X)}\right].
\end{align*}
\subsection{Canonical gradient}
In combination with Section \ref{ssec:cal-S1}-\ref{ssec:cal-S3}, 
we have  $\E[\{\varphi(O)-\tau\}S_{\obs}(Z,X,V,Y,R)]=0$, where 
\begin{align*} \varphi(O)&=\frac R{\rho(X)}\left[\frac Z{\pi(X,V)}\left\{Y-\mu(1,X,V)\right\}-\frac{1-Z}{1-\pi(X,V)}\left\{Y-\mu(0,X,V)\right\}+m(X,V)\right]\\&~~~~-\frac{\left\{R-\rho(X)\right\}\E\left\{m(X,V)\mid X\right\}}{\rho(X)} -\tau { .}
\end{align*}
  Because $\varphi(O)-\tau$ lies in the tangent space, it is the EIF for $\tau$. 
 \end{proof}

 \section{The proof of Theorem \ref{THM: MUL-ROBUSTNESS}}
 \subsection{Proof of the triple robustness}
 In this section, let $h(X)=\E\{m(X,V)\mid X\}$, where $m(X,V)=\mu(1,X,V)-\mu(0,X,V)$.  
 
First, suppose that $\pi(X,V)$ and ${\rho}( X)$ are correctly specified; 
{ $m^\ast(X,V)$} and $f^\ast(V\mid X)$ are misspecified parametric models, and
\begin{align*} &\E\left(\frac R{{\rho}(X)}\left[\frac Z{\pi(X,V)}\left\{Y-\mu^\ast(1,X,V)\right\}-\frac{1-Z}{1-\pi(X,V)}\left\{Y-\mu^\ast(0,X,V)\right\}\right]\right)\\&\;\;\;\;\;\;+\E\left[\frac{R\left\{{ m^\ast(X,V)}-h^\ast(X)\right\}}{{\rho}(X)}+h^\ast(X)\right]\\& =\E\left(\E\left[\left.\left\{\mu(1,X,V)-\mu^\ast(1,X,V)\right\}-\left\{\mu(0,X,V)-\mu^\ast(1,X,V)\right\}\right|R=1,X\right]\right)\\&\;\;\;\;\;\;+\E\left[\frac{R\left\{{ m^\ast(X,V)}-h^\ast(X)\right\}}{{\rho}(X)}+h^\ast(X)\right]\\& =\E\left(\E\left[\left.\left\{\mu(1,X,V)-\mu^\ast(1,X,V)\right\}-\left\{\mu(0,X,V)-\mu^\ast(1,X,V)\right\}\right|R=1,X\right]\right)\\&\;\;\;\;\;\;+\E\left[\frac{R\left\{{ m^\ast(X,V)}-h^\ast(X)\right\}}{{\rho}(X)}+h^\ast(X)\right]\\& =\E\left(\E\left\{\mu(1,X,V)-\mu(0,X,V)\mid  R=1,X\right\}\right) -\E\left(\E\left\{{ m^\ast(X,V)}\mid  R=1,X\right\}\right)\\&\;\;\;\;\;\;+\E\left\{\frac{R{ m^\ast(X,V)}}{{\rho}(X)}\right\}\\&=\tau. 
\end{align*}
Second, suppose that $\mu(Z,X,V)$  and ${\rho}(X)$ are correctly specified; 
$f^\ast(V\mid X)$ and $\pi^\ast(X,V)$ are misspecified parametric models. Thus  $m(X,V)=\mu(1, X,V)-\mu(0, X,V)$ is correctly specified, and
\begin{align*}
 \E&\left(\frac R{{\rho}(X)}\left[\frac Z{\pi^\ast(X,V)}\left\{Y-\mu(1,X,V)\right\}-\frac{1-Z}{1-\pi^\ast(X,V)}\left\{Y-\mu(0,X,V)\right\}\right]\right.\\&\;\;\;\;\;\;\;\;\;\left.+\frac R{{\rho}(X)}\left\{m(X,V)-h^\ast(X)\right\}+h^\ast(X) \right)\\& =\E\left(\frac{{\rho}(X)}{{\rho}(X)}\left[\frac{\pi(X,V)}{\pi^\ast(X,V)}\left\{\mu(1,X,V)-\mu(1,X,V)\right\}+\frac{1-\pi(X,V)}{1-\pi^\ast(X,V)}\left\{\mu(0,X,V)-\mu(0,X,V)\right\}\right]\right.\\&\;\;\;\;\;\;\;\;\;\left.+\frac R{{\rho}(X)}\left\{m(X,V)-h^\ast(X)\right\}+h^\ast(X) \right)\\& =\E\left\{\frac{Rm(X,V)}{{\rho}(X)}\right\}\\&=\tau. 
\end{align*}
Third, suppose that $\mu(Z,X,V)$ and $f(V\mid X)$ are correctly specified; 
{ $\rho^\ast(X)$} and $\pi^\ast(X)$ are misspecified parametric models. Thus  $m(X,V)=\mu(1, X,V)-\mu(0, X,V)$ is correctly specified, and
\begin{align*}
 &\E\left(\frac R{{{ \rho}}^\ast(X)}\left[ \frac Z{\pi^\ast(X,V)}\left\{Y-\mu(1,X,V)\right\}-\frac{1-Z}{1-\pi^\ast(X,V)}\left\{Y-\mu(0,X,V)\right\}\right]\right.\\&\;\;\;\;\;\;\;\;\;\left.+\frac R{{{ \rho}}^\ast(X)}\left\{m(X,V)-h(X)\right\}+h(X)\right)\\& =\E\left(\frac{{\rho}(X)}{{{ \rho}}^\ast(X)}\left[\frac{\pi(X,V)}{\pi^\ast(X,V)}\left\{\mu(1,X,V)-\mu(1,X,V)\right\}+\frac{1-\pi(X,V)}{1-\pi^\ast(X,V)}\left\{\mu(0,X,V)-\mu(0,X,V)\right\}\right]\right)\\&\;\;\;\;\;+\E\left[\frac{{\rho}(X)\left\{m(X,V)-h(X)\right\}}{{{ \rho}}^\ast(X)}\right]+\E\left\{h(X)\right\}\\&=0+\E\left[\frac{{\rho}(X)\E\left\{m(X,V)-h(X)\mid X\right\}}{{{ \rho}}^\ast(X)}\right]+\E\left\{h(X)\right\}\\&=\tau. 
\end{align*} 

 \subsection{Proof of the semiparamteric efficiency} 
 Under the intersection submodel $\cap_{j=1}^3\mathcal{M}_j$,   all nuisance models are correctly specified. 
 Therefore, by the Central Limit Theorem,  we know that $\hat{\tau}_\MR$ achieves the local efficiency.
 
\section{The proof of Theorem \ref{THM: NON-MARCHINE-LEARNING}}
 {We show that $\widehat{\tau}_{\MR}$ using $\{\rho(X;\widehat{\alpha}), \pi(X,V;\widehat\beta),\mu(Z,X,V;\widehat \gamma),f(V\mid X;\widehat \theta)\}$ satisfying the regularity conditions (a)-(d) in Theorem  \ref{THM: NON-MARCHINE-LEARNING} is asymptotic normality and has the influence function ${\psi_{\mathrm{eff}}^\tau}$ as in Theorem \ref{THM: EIF}. Therefore,   it can achieve the semiparametric efficiency. 
 For simplicity, let $\{\widehat \rho(X), \widehat \pi(X,V),\widehat \mu(Z,X,V ),\widehat f(V\mid X )\}$ denote the estimated nuisance functions  $\{\rho(X;\widehat \alpha), \pi(X,V;\widehat \beta),\mu(Z,X,V ;\widehat \gamma),f(V\mid X ;\widehat \theta)\}$. Let $(\alpha^\star,\beta^\star,\gamma^\star,\theta^\star)$ denote the   probability limit of  $( \widehat\alpha, \widehat\beta, \widehat\gamma, \widehat\theta)$. Let $(\alpha^\ast,\beta^\ast,\gamma^\ast,\theta^\ast)$ denote the   true value,  and we have $\{ \rho(X),  \pi(X,V), \mu(Z,X,V ), f(V\mid X )\}=\{\rho(X; \alpha^\ast), \pi(X,V; \beta^\ast),\mu(Z,X,V ; \gamma^\ast),f(V\mid X ; \theta^\ast)\}$.    We use $\mathbb{P}_n$ to denote the empirical measure so that sample averages are written as $\mathbb{P}_n(f)=\mathbb{P}_n\{f(U)\}=\frac{1}{n} \sum_i f\left(U_i\right)$ for some generic variable $U$. For a possibly random function $\widehat{f}$,  we similarly write $\mathbb{P}(\widehat{f})=\mathbb{P}\{\widehat{f}(U)\}=\int \widehat{f}(u) d \mathbb{P}(u)$, and   $\mathbb{P}_{\widehat \theta}(\widehat{f})=\mathbb{P}\{\widehat{f}(U)\}=\int \widehat{f}(u) d \mathbb{P}(u;\hat{\theta})$. 
 In the following proof, we will use the following identities repeatedly.   
{ \begin{align*}  
   \mathbb{P}_{\widehat{\theta}}\left\{\widehat{\mu}(1,X,V)\mid   X\right\}& =\int{\widehat{\mu}(1,X,V)}\widehat{f}\left(V\mid X\right)d\mu\left(V\right)\\&=\int\frac{\widehat{f}\left(V\mid X\right)}{f\left(V\mid R=1,X\right)}{\widehat{\mu}(1,X,V)}f\left(V\mid R=1,X\right)d\mu\left(V\right)\\& ={ \P }\left\{\left.\frac{\widehat{f}\left(V\mid X \right)}{f\left(V\mid X \right)}{\widehat{\mu}(1,X,V)}\right|R=1,X\right\}, \\
\P\left[\P_{\widehat{\theta}}\left\{\widehat{\mu}(1,X,V)\mid X\right\}\right]& =\mathbb{P}\left[{ \P }\left\{\left.\frac{\widehat{f}(V\mid X )}{f (V\mid X )}{\widehat{\mu}(1,X,V)}\right|R=1,X\right\}\right]\\& =\P\left\{\frac R{{\rho}(X)}\frac{\widehat{f}(V\mid X )}{f (V\mid X )}{\widehat{\mu}(1,X,V)}\right\}.\\&\\\mathbb{P}\left[\frac R{\widehat{\rho}(X)}\P_{\widehat{\theta}}\left\{\widehat{\mu}(1,X,V)\mid X\right\}\right]& =\mathbb{P}\left[{ \P }\left\{\left.\frac R{\widehat{\rho}(X)}\frac{\widehat{f}(V\mid X )}{f (V\mid X )}{\widehat{\mu}(1,X,V)}\right|R=1,X\right\}\right]\\& =\mathbb{P}\left\{\frac R{\widehat{\rho}(X)}\frac{\widehat{f}(V\mid X )}{f (V\mid X )}{\widehat{\mu}(1,X,V)}\right\}.
\end{align*}
 Next, we { prove} Theorem \ref{THM: NON-MARCHINE-LEARNING}.}
	\begin{proof}
 Let $\mathbb{P}$ denote the true distribution, $\mathbb{P}_n$ denote the sample averages,  and
 \begin{equation}
     \label{eq: mr-estimator-1}
     \begin{aligned}
 \tau_{\MR ,0}(\alpha,\beta,\gamma,\theta)&=  \frac R{\rho(X;{\alpha})}\left[\frac Z{\pi(X,V;\beta)}\left\{Y-\mu(1,X,V;\gamma)\right\}+\mu(1,X,V;\gamma) \right]  \\&\;\;\;
+\frac {\rho(X;{\alpha})-R}{\rho(X;{\alpha})}  \E_\theta\{\mu(1,X,V;\gamma)\mid X\}   ,
\end{aligned}
 \end{equation} 
 \begin{equation}
     \label{eq: mr-estimator-0}
     \begin{aligned}
\tau_{\MR ,0}(\alpha,\beta,\gamma,\theta)&=   \frac R{\rho(X;{\alpha})}\left[\frac{1-Z}{1-\pi(X,V;\beta)}\left\{Y-\mu(0,X,V;\gamma)\right\}+\mu(0,X,V;\gamma ) \right] \\&\;\;\;
+\frac {\rho(X;{\alpha})-R}{\rho(X;{\alpha})} \E_\theta\{\mu(0,X,V;\gamma)\mid X\}   ,
\end{aligned}
 \end{equation}
 and $\tau_{\MR }(\alpha,\beta,\gamma,\theta)=\tau_{\MR ,1}(\alpha,\beta,\gamma,\theta)-\tau_{\MR,0 }(\alpha,\beta,\gamma,\theta)$. 
{Condition (a) implies that $ (\alpha^\star,\beta^\star,\gamma^\star,\theta^\star)\\=(\alpha^\ast,\beta^\ast,\gamma^\ast,\theta^\ast)$   and thus $\tau_{\MR}(\alpha^\star,\beta^\star,\gamma^\star,\theta^\star)=\tau_{\MR}(\alpha^\ast,\beta^\ast,\gamma^\ast,\theta^\ast)$.

Note that
\begin{align}
\mathbb{P}_n &\notag\{\tau_{\MR,1}( \widehat\alpha, \widehat\beta, \widehat\gamma, \widehat\theta)\}- \tau_1  \\&\notag=\left(\mathbb{P}_n-\mathbb{P}\right)\{\tau_{\MR,1}( \widehat\alpha, \widehat\beta, \widehat\gamma, \widehat\theta)\}+\left[\mathbb{P}\left\{\tau_{\MR,1}( \widehat\alpha, \widehat\beta, \widehat\gamma, \widehat\theta)\right\}- \tau_1 \right] \\
&\notag=\left(\mathbb{P}_n-\mathbb{P}\right)\{\tau_{\MR,1}( \alpha^\ast,\beta^\ast,\gamma^\ast,\theta^\ast)\}+\left[\mathbb{P}
\left\{\tau_{\MR,1}( \widehat\alpha, \widehat\beta, \widehat\gamma, \widehat\theta)\right\}- \tau_1 \right]\\
&\notag~~~~+\left(\mathbb{P}_n-\mathbb{P}\right)\big\{\tau_{\MR,1}( \widehat\alpha, \widehat\beta, \widehat\gamma, \widehat\theta)-\tau_{\MR,1}( \alpha^\ast,\beta^\ast,\gamma^\ast,\theta^\ast)\big\}\\
& =\left(\mathbb{P}_n-\mathbb{P}\right)\{\tau_{\MR,1}( \alpha^\ast,\beta^\ast,\gamma^\ast,\theta^\ast)\}+\left[\mathbb{P}
\left\{\tau_{\MR,1}( \widehat\alpha, \widehat\beta, \widehat\gamma, \widehat\theta)\right\}- \tau_1 \right]+o_p \left(n^{-1 / 2}\right), 
&\label{eq:normality} 
\end{align}
where the third equality follows by Condition (b) and the empirical process theory \citep{van2000asymptotic}.  The first term
\begin{align*}
 \left(\mathbb{P}_n-\mathbb{P}\right)\{\tau_{\MR,1}( \alpha^\ast,\beta^\ast,\gamma^\ast,\theta^\ast)\}
\end{align*}
is a simple sample average of a fixed function, and so by the central limit theorem,  it behaves as a normally distributed random variable with variance $\operatorname{var}(\varphi) / n$, up to error $o_p(n^{-1/2})$.   It remains to analyze the second term,}	 
\begin{align*} 
&{\mathbb{P}}\left\{\tau_{\MR,1}( \widehat\alpha, \widehat\beta, \widehat\gamma, \widehat\theta)- \tau_1  \right\}\\&  =\mathbb{P}\left[\frac R{\widehat{\rho}(X)}\frac Z{\widehat{\pi}(X,V)}\left\{Y-\widehat{\mu}(1,X,V)\right\}+\frac R{\widehat{\rho}(X)}\left\{\widehat{\mu}(1,X,V)-\frac{\widehat{f} (V\mid X )}{f (V\mid X )}\widehat{\mu}(1,X,V)\right\}\right.\\&\left.\;\;\;\;\;\;\;\;\;\;\;+\frac R{{\rho}(X)}\frac{\widehat{f} (V\mid X )}{f (V\mid X )}\widehat{\mu}(1,X,V)\right]-\mathbb{P}\left\{\frac R{{\rho}(X)}\mu(1,X,V)\right\}\\& =\mathbb{P}\left[\frac R{\widehat{\rho}(X)}\frac Z{\widehat{\pi}(X,V)}\left\{Y-\widehat{\mu}(1,X,V)\right\}+\frac R{\widehat{\rho}(X)}\left\{\widehat{\mu}(1,X,V)-\frac{\widehat{f} (V\mid X )}{f (V\mid X )}\widehat{\mu}(1,X,V)\right\}\right]\\&\;\;\;\;\;\;\;\;\;\;\;\;+\mathbb{P}\left[\frac R{{\rho}(X)}\left\{\frac{\widehat{f} (V\mid X )}{f (V\mid X )}\widehat{\mu}(1,X,V)-\mu\left(1,X,V\right)\right\}\right]\\& =\mathbb{P}\left[\frac R{\widehat{\rho}(X)}\left\{\frac{ZY}{\widehat{\pi}(X,V)}-\frac{Z\widehat{\mu}(1,X,V)}{\widehat{\pi}(X,V)}+\widehat{\mu}(1,X,V)-\frac{\widehat{f} (V\mid X )}{f (V\mid X )}\widehat{\mu}\left(1,X,V\right)\right\}\right]\\&\;\;\;\;\;\;\;\;\;\;\;\;+\mathbb{P}\left[\frac R{{\rho}(X)}\left\{\frac{\widehat{f} (V\mid X )}{f (V\mid X )}\widehat{\mu}(1,X,V)-\mu\left(1,X,V\right)\right\}\right]\\& =\mathbb{P}\left[\frac R{\widehat{\rho}(X)}\left\{\frac{ZY}{\widehat{\pi}(X,V)}-\frac{Z\widehat{\mu}(1,X,V)}{\widehat{\pi}(X,V)}+\widehat{\mu}(1,X,V)-\frac{\widehat{f} (V\mid X )}{f (V\mid X )}\widehat{\mu}(1,X,V)\right\}\right]\\&\;\;\;\;\;\;\;\;\;\;\;\;-\mathbb{P}\left[\frac{R}{\rho(X)}\left\{\frac{ZY}{\widehat{\pi}(X,V)}-\frac{Z\widehat{\mu}(1,X,V)}{\widehat{\pi}(X,V)}+\widehat{\mu}(1,X,V)-\frac{\widehat{f} (V\mid X )}{f (V\mid X )}\widehat{\mu}(1,X,V)\right\}\right]\\&\;\;\;\;\;\;\;\;\;\;\;\;+\mathbb{P}\left[\frac{R}{\rho(X)}\left\{\frac{ZY}{\widehat{\pi}(X,V)}-\frac{Z\widehat{\mu}(1,X,V)}{\widehat{\pi}(X,V)}+\widehat{\mu}(1,X,V)-\frac{\widehat{f} (V\mid X )}{f (V\mid X )}\widehat{\mu}(1,X,V)\right\}\right]\\&\;\;\;\;\;\;\;\;\;\;\;\;+\mathbb{P}\left[\frac R{{\rho}(X)}\left\{\frac{\widehat{f} (V\mid X )}{f (V\mid X )}\widehat{\mu}(1,X,V)-\mu(1,X,V)\right\}\right]\\&  =\mathbb{P}\left[\frac R{\widehat{\rho}(X)}\left\{\frac{ZY}{\widehat{\pi}(X,V)}-\frac{Z\widehat{\mu}(1,X,V)}{\widehat{\pi}(X,V)}+\widehat{\mu}(1,X,V)-\frac{\widehat{f} (V\mid X )}{f (V\mid X )}\widehat{\mu}(1,X,V)\right\}\right]\\&\;\;\;\;\;\;\;\;\;\;\;\;-\mathbb{P}\left[\frac{R}{\rho(X)}\left\{\frac{ZY}{\widehat{\pi}(X,V)}-\frac{Z\widehat{\mu}(1,X,V)}{\widehat{\pi}(X,V)}+\widehat{\mu}(1,X,V)-\frac{\widehat{f} (V\mid X )}{f (V\mid X )}\widehat{\mu}(1,X,V)\right\}\right]\\&\;\;\;\;\;\;\;\;\;\;\;\;+\mathbb{P}\left[\frac{R}{\rho(X)}\left\{\frac{ZY}{\widehat{\pi}(X,V)}-\frac{Z\widehat{\mu}(1,X,V)}{\widehat{\pi}(X,V)}+\widehat{\mu}(1,X,V)-\mu(1,X,V)\right\}\right] \\& =\mathbb{P}\left[\left\{\frac{R}{{ \widehat{\rho}(X)}}-\frac{R}{{\rho(X)}}\right\}\begin{Bmatrix}
  \dfrac{\pi(X,V)\mu(1,X,V)}{\widehat{\pi}(X,V)}-\dfrac{\pi(X,V)\widehat{\mu}(1,X,V)}{\widehat{\pi}(X,V)}\\\addlinespace[1.5mm]+\widehat{\mu}(1,X,V)-\dfrac{\widehat{f} (V\mid X )}{f (V\mid X )}\widehat{\mu}(1,X,V)
\end{Bmatrix} \right]\\&\;\;\;\;\;\;\;\;\;\;\;\;+\mathbb{P}\left[\frac{R}{\rho(X)}\left\{\dfrac{\pi(X,V)\mu\left(1,X,V\right)}{\widehat{\pi}(X,V)}-\dfrac{\pi(X,V)\widehat{\mu}(1,X,V)}{\widehat{\pi}(X,V)}+\widehat{\mu}\left(1,X,V\right)-\mu\left(1,X,V\right)\right\}\right]\\&  =\mathbb{P}\left[\left\{\frac{R}{{ \widehat{\rho}(X)}}-\frac{R}{{\rho(X)}}\right\}\left\{\frac{\pi(X,V)\mu(1,X,V)}{\widehat{\pi}(X,V)}-\frac{\pi(X,V)\widehat{\mu}(1,X,V)}{\widehat{\pi}(X,V)} \right\}\right]\\& \;\;\;\;\;\;\;\;\;\;\;\;+\mathbb{P}\left[\left\{\frac{R}{{ \widehat{\rho}(X)}}-\frac{R}{{\rho(X)}}\right\}\left\{\widehat{\mu}(1,X,V)-\dfrac{\widehat{f} (V\mid X )}{f (V\mid X )}\widehat{\mu}(1,X,V)\right\}\right] \\&\;\;\;\;\;\;\;\;\;\;\;\;+\mathbb{P}\left[\frac{R}{\rho(X)}\left\{\frac{\pi(X,V)}{\widehat{\pi}(X,V)}\left\{\mu\left(1,X,V\right)-\widehat{\mu}(1,X,V)\right\}+\widehat{\mu}\left(1,X,V\right)-\mu\left(1,X,V\right)\right\}\right]\\& =\mathbb{P}\left[\left\{\frac{R}{{ \widehat{\rho}(X)}}-\frac{R}{{\rho(X)}}\right\}\left\{\frac{\pi(X,V)\mu(1,X,V)}{\widehat{\pi}(X,V)}-\frac{\pi(X,V)\widehat{\mu}(1,X,V)}{\widehat{\pi}(X,V)}\right\}\right]\\&\;\;\;\;\;\;\;\;\;\;\;\;+\mathbb{P}\left[\left\{\frac{R}{{ \widehat{\rho}(X)}}-\frac{R}{{\rho(X)}}\right\}\left\{\widehat{\mu}\left(1,X,V\right)-\frac{\widehat{f} (V\mid X )}{f (V\mid X )}\widehat{\mu}\left(1,X,V\right)\right\}\right]\\&\;\;\;\;\;\;\;\;\;\;\;\;+\mathbb{P}\left\{\frac{R}{{\rho}(X)}\left\{\frac{\pi(X,V)}{\widehat{\pi}(X,V)}-1\right\}\left\{\mu\left(1,X,V\right)-\widehat{\mu}(1,X,V)\right\}\right\}\\& =\mathbb{P}\left[\frac{\pi(X,V)R}{\widehat{\pi}(X,V) }\left\{\frac{1}{{ \widehat{\rho}(X)}}-\frac{1}{{{\rho}(X)}}\right\}\left\{\mu(1,X,V)-\widehat{\mu}(1,X,V)\right\}\right]\\&\;\;\;\;\;\;\;\;\;\;\;\;+\mathbb{P}\left[\frac{\widehat{\mu}\left(1,X,V\right)R}{f (V\mid X ) }\left\{\frac{1}{{ \widehat{\rho}(X)}}-\frac{1}{{{\rho}(X)}}\right\}\left\{f (V\mid X )-\widehat{f} (V\mid X )\right\}\right]\\&\;\;\;\;\;\;\;\;\;\;\;\;+\mathbb{P}\left[\frac R{\widehat{\pi}(X,V)\rho(X)}\left\{\pi(X,V)-\widehat{\pi}(X,V)\right\}\left\{\mu\left(1,X,V\right)-\widehat{\mu}(1,X,V)\right\}\right]\\& =\mathbb{P}\left[\frac{\pi(X,V)R}{\widehat{\pi}(X,V) {{\rho}(X)}{\widehat{\rho}(X)}}\left\{ {{{\rho}(X)}}-{\widehat{{\rho}}(X)}
\right\}\left\{\mu(1,X,V)-\widehat{\mu}(1,X,V)\right\}\right]\\&\;\;\;\;\;\;\;\;\;\;\;\;+\mathbb{P}\left[\frac{\widehat{\mu}\left(1,X,V\right)R}{f (V\mid X ) {{\rho}(X)}{\widehat{\rho}(X)}}\left\{ {{{\rho}(X)}}-{\widehat{{\rho}}(X)}
\right\}\left\{f (V\mid X )-\widehat{f} (V\mid X )\right\}\right]\\&\;\;\;\;\;\;\;\;\;\;\;\;+\mathbb{P}\left[\frac R{\widehat{\pi}(X,V)\rho(X)}\left\{\pi(X,V)-\widehat{\pi}(X,V)\right\}\left\{\mu\left(1,X,V\right)-\widehat{\mu}(1,X,V)\right\}\right].
\end{align*}
Similarly,
\begin{align*}
   &{ \P}\left\{\tau_{\MR,0}( \widehat\alpha, \widehat\beta, \widehat\gamma, \widehat\theta)-\tau_{ 0} \right\}\\& =\mathbb{P}\left[\frac{\{1-\pi(X,V)\}R}{\{1-\widehat{\pi}(X,V) \}{{\rho}(X)}{\widehat{\rho}(X)}}\left\{ {{{\rho}(X)}}-{\widehat{{\rho}}(X)}
\right\}\left\{\mu(0,X,V)-\widehat{\mu}(0,X,V)\right\}\right]\\&\;\;\;\;\;\;\;\;\;\;\;\;+\mathbb{P}\left[\frac{\widehat{\mu}\left(0,X,V\right)R}{f (V\mid X ){{\rho}(X)}{\widehat{\rho}(X)}}\left\{ {{{\rho}(X)}}-{\widehat{{\rho}}(X)}
\right\}\left\{f (V\mid X )-\widehat{f} (V\mid X )\right\}\right]\\&\;\;\;\;\;\;\;\;\;\;\;\;-\mathbb{P}\left[\frac R{\{1-\widehat{\pi}(X,V)\}\rho(X)}\left\{\pi(X,V)-\widehat{\pi}(X,V)\right\}\left\{\mu\left(0,X,V\right)-\widehat{\mu}(0,X,V)\right\}\right].
\end{align*}

By the Cauchy–Schwarz inequality and Conditions (a), (c) and (d), it follows that for  a  constant $C$, we have \begin{align*} \left|\mathbb{P}\left\{\tau_{\MR}(\widehat\alpha,\widehat\beta,\widehat\gamma,\widehat\theta)\right\}-\tau \right| &\leq C{\left\|\rho(X)-\widehat \rho(X)\right\|}_2{\left\|m(X,V)-\widehat m(X,V)\right\|}_2\\&\;\;\;\;+C{\left\|\rho(X)-\widehat \rho(X)\right\|}_2{\left\|f(V\mid X)-\widehat f(V\mid X)\right\|}_2\\&\;\;\;\;+C{\left\|\pi(X,V)-\widehat\pi(X,V)\right\|}_2{\left\|m(X,V)-\widehat m(X,V)\right\|}_2\\&\;\;\;\;=o_p(n^{-1/2}).
\end{align*} 
Continuing with \eqref{eq:normality}, we have 
\begin{align*}
\mathbb{P}_n\notag \{\tau_{\MR,1}&( \widehat\alpha, \widehat\beta, \widehat\gamma, \widehat\theta)\}-\tau_1=\left(\mathbb{P}_n-\mathbb{P}\right)\{\tau_{\MR,1}( \alpha^\ast,\beta^\ast,\gamma^\ast,\theta^\ast)\}+ o_p \left(n^{-1 / 2}\right). 
\end{align*} 
This completes the proof.
	\end{proof}
 \section{Derivation of the optimal sample allocation}  
 \subsection{The influence function of $\widehat{\tau}_{\MR} (\widehat{\beta}, \widehat{\gamma}, \widehat{\theta})$}
 For the selection probability, we consider the  estimating equation   for $\alpha$, 
 \begin{align}
\label{eq:rho-ee}
\E\{S_\rho(R, X ; \alpha)\}=0,
\end{align}
where $$ 
{\ S_\rho(R, X ; \alpha)} =\frac{R-\rho(X;\alpha)}{\rho(X;\alpha)\{1-\rho(X;\alpha)\}} \frac{\partial \rho(X;\alpha)}{\partial \alpha}.$$   
 For the propensity score, we consider the estimating equation  for $\beta$,  
 \begin{align}
\label{eq:pi-ee}
\E\{S_\pi(Z, X, V ; \beta)\mid R=1\}=0,
\end{align}
where
 $$ 
{\ S_\pi(Z, X, V ; \beta)} =\frac{Z-\pi(X, V ; \beta)}{\pi(X, V ; \beta)\{1-\pi(X, V ; \beta)\}} \frac{\partial \pi(X, V ; \beta)}{\partial \beta}.$$  
  For the outcome model, we consider the following  estimating function for $\gamma$, 
\begin{align}
\label{eq:mu-ee}
 \E\left\{{  S_\mu(Z, X, V, Y ; \gamma)}\mid R=1\right\}=0,
\end{align}
  where $${  S_\mu(Z, X, V, Y ; \gamma)}=  {\frac{\partial \mu(Z,X, V ; \gamma)}{\partial \gamma}\left\{Y-\mu(Z,X,V ; \gamma)\right\}}.$$ For the imputation model, we consider the following estimating equation for  $\theta$,  
\begin{align}
    \label{eq:den-ee}
 \E\left\{{  S_f( X, V  ; \theta)}\mid R=1\right\}=0,
\end{align}
where  $$    S_f(X, V ; \theta)= \dfrac{\partial\log f\left(  V\mid X,R=1 ; \theta\right) }{\partial\theta}.$$
  In addition, let  $\widehat{\alpha}$, $\widehat{\beta} $, $\widehat{\gamma} $ and $\widehat{\theta} $ be the estimators solving the corresponding empirical estimating equations based on linked cohort, with probability limits $\alpha^\star$,  $\beta^\star $,  $\gamma^\star$ and  ${\theta^\star} $,  respectively.
  
\begin{proposition}
Under Assumptions~\ref{ass:ignorability}-\ref{eq: indep-R},    the influence function of  $ \widehat{\tau}_{\MR} ( \widehat{\beta}, \widehat{\gamma}, \widehat{\theta}) $ is given by
\begin{align*}
\begin{aligned}
\psi_{\MR}^{\tau}(O)
 & ={S_1(\beta^\star, \gamma^\star, \theta^\star)}+\frac{R}{\rho} S_2(\beta^\star, \gamma^\star, \theta^\star),
\end{aligned}
\end{align*} 
where  
\begin{gather*}S_1(\beta,\gamma,\theta)  =  \delta(1,X;\gamma,\theta)-\delta(0,X;\gamma,\theta)-\tau,\\
    \begin{aligned}
     S_2(\beta,\gamma,\theta)&= \begin{bmatrix}
    \dfrac { Z}{ \pi(X,V;\beta)  }\left\{Y-\mu(1,X,V;\gamma)\right\}+\mu(1,X,V;\gamma)-\delta(1,X;\gamma,\theta) \\\addlinespace[1.5mm]-\dfrac { 1-Z}{ 1-\pi(X,V;\beta)  }\left\{Y-\mu(0,X,V;\gamma)\right\}-\mu(0,X,V;\gamma)+\delta( 0,X;\gamma,\theta)  
\end{bmatrix}\\&~~~~~~ + \E\left\{\frac{\partial \widehat{\tau}_{\MR}\left(\beta^\star, \gamma^\star, \theta^\star\right)}{\partial \beta }\right\} \E\left\{\frac{\partial S_\pi\left(Z, X, V; \beta^\star\right)}{\partial \beta }\right\}^{-1} S_{\pi  }(Z, X, V ; \beta)  \\&~~~~~~+ \E\left\{\frac{\partial \widehat{\tau}_{\MR}\left(\beta^\star, \gamma^\star, \theta^\star\right)}{\partial \gamma }\right\} \E\left\{\frac{\partial S_\mu\left(Z, X, V, Y ; \gamma^\star\right)}{\partial \gamma }\right\}^{-1}S_{\mu  }(Z, X, V,Y ; \gamma),
\end{aligned} \\   \frac{\partial \widehat{\tau}_{\MR}\left(\beta^\star, \gamma^\star, \theta^\star\right)}{\partial \beta } = \frac{R}{\rho}\begin{Bmatrix}
    \dfrac { Z(Y-\mu_1^\star )\Dot{\pi}^\star}{ (\pi^\star ) ^2}    -\dfrac { (1-Z)(Y-\mu_1^\star )\Dot{\pi}^\star}{ (1-\pi^\star)^2
  } 
\end{Bmatrix}, \\\frac{\partial \widehat{\tau}_{\MR}\left(\beta^\star, \gamma^\star, \theta^\star\right)}{\partial \gamma } 
=\frac{R}{\rho}\left\{ \left(1-
    \dfrac { Z}{ \pi^\star  }\right) \Dot{\mu}_1^\star- \Dot{\delta}_{1,\gamma}^\star-\left(1-
    \dfrac { 1-Z}{ 1-\pi^\star  }\right) \Dot{\mu}_0^\star+ \Dot{\delta}_{0,\gamma}^\star\right\}+ \Dot{\delta}_{1,\gamma}^\star-  \Dot{\delta}_{0,\gamma}^\star.
\end{gather*}   
  $$\pi^\star=\pi\left(X, V ; \beta^\star\right), ~\dot{\pi}^\star=\frac{\partial \pi\left(X, V ; \beta^\star\right) }{\partial \beta} ,~ S_{\pi  }^\star= \frac{R S_\pi(Z, X, V ; \gamma^\star)}{\rho}.$$  $$\mu_{z }^\star=\mu \left(z,X, V ; \gamma^\star\right), ~\dot{\mu}_{z }^\star=\frac{\partial \mu_z\left(X, V ; \gamma^\star\right) }{\partial \gamma },~S_{\mu  }^\star=\frac{RS_\mu\left(Z, X, V, Y ; \gamma^\star\right)}{\rho}.$$  $$\dot{\delta}_{z,\gamma  }^\star=\frac{\partial \delta_z\left(X, V ; \gamma^\star, \theta^\star\right)}{\partial \gamma} ,~\dot{\delta}_{z,\theta}^\star=\frac{\partial \delta_z\left(X, V ; \gamma^\star,\theta^\star\right)}{\partial \theta} .$$  
\end{proposition} 
\begin{proof} 
  We write $\widehat{\tau}_{\MR}=\widehat{\tau}_{\MR}(\widehat{\beta}, \widehat{\gamma}, \widehat{\theta})$ to emphasize its dependence on the parameter estimates $(\widehat{\beta}, \widehat{\gamma}, \widehat{\theta})$. By the Taylor expansion, 
\begin{equation*}
    \begin{aligned}
\widehat{\tau}_{\MR}&(\widehat{\beta}, \widehat{\gamma}, \widehat{\theta})\\
\cong & \widehat{\tau}_{\MR}\left(\beta^\star, \gamma^\star, \theta^\star\right)+\E\left\{\frac{\partial \widehat{\tau}_{\MR}\left(\beta^\star, \gamma^\star, \theta^\star\right)}{\partial \beta }\right\}\left(\widehat{\beta}-\beta^\star\right) \\
& +\E\left\{\frac{\partial \widehat{\tau}_{\MR}\left(\beta^\star, \gamma^\star, \theta^\star\right)}{\partial \gamma^\T}\right\}\left(\widehat{\gamma} -\gamma^\star\right) +\E\left\{\frac{\partial \widehat{\tau}_{\MR}\left(\beta^\star, \gamma^\star, \theta^\star\right)}{\partial \theta^\T}\right\}\left(\widehat{\theta} -\theta^\star\right)~~~~~~~~~~~~~~~~~~~~~\\\cong & \widehat{\tau}_{\MR}\left(\beta^\star, \gamma^\star, \theta^\star\right) \\
& +n ^{-1} \sum_{ j } \E\left\{\frac{\partial \widehat{\tau}_{\MR}\left(\beta^\star, \gamma^\star, \theta^\star\right)}{\partial \beta }\right\} \E\left\{\frac{\partial S_\pi\left(Z, X, V; \beta^\star\right)}{\partial \beta }\right\}^{-1} { \frac{R}{\rho}}S_\pi\left(Z_j, X_j, V_j ; \beta^\star\right)\\
& +n ^{-1} \sum_{ j } \E\left\{\frac{\partial \widehat{\tau}_{\MR}\left(\beta^\star, \gamma^\star, \theta^\star\right)}{\partial \gamma^\T}\right\} \E\left\{\frac{\partial S_\mu\left(Z, X, V, Y ; \gamma^\star\right)}{\partial \gamma^\T}\right\}^{-1}{ \frac{R}{\rho}} S_\mu\left(Z_j, X_j, V_j, Y_j ; \gamma^\star\right) \\
& +n ^{-1} \sum_{ j } \E\left\{\frac{\partial \widehat{\tau}_{\MR}\left(\beta^\star, \gamma^\star, \theta^\star\right)}{\partial \theta^\T}\right\} \E\left\{\frac{\partial S_f\left( X, V ; \theta^\star\right)}{\partial \theta^\T}\right\}^{-1} { \frac{R}{\rho}}S_f\left(  X_j, V_j  ; \theta^\star\right)  { .}
\end{aligned}
\end{equation*}
where
  \begin{gather*}
     \frac{\partial \widehat{\tau}_{\MR}\left(\beta^\star, \gamma^\star, \theta^\star\right)}{\partial \beta^\T}= \frac{R}{\rho}\begin{Bmatrix}
    \dfrac { Z(Y-\mu_1^\star )\Dot{\pi}^\star}{ (\pi^\star ) ^2}    -\dfrac { (1-Z)(Y-\mu_1^\star )\Dot{\pi}^\star}{ (1-\pi^\star)^2
  } 
\end{Bmatrix} {,}\\
     \frac{\partial \widehat{\tau}_{\MR}\left(\beta^\star, \gamma^\star, \theta^\star\right)}{\partial \gamma^\T}=\frac{R}{\rho}\left\{ \left(1-
    \dfrac { Z}{ \pi^\star  }\right) \Dot{\mu}_1^\star- \Dot{\delta}_{1,\gamma}^\star-\left(1-
    \dfrac { 1-Z}{ 1-\pi^\star  }\right) \Dot{\mu}_0^\star+ \Dot{\delta}_{0,\gamma}^\star\right\}+ \Dot{\delta}_{1,\gamma}^\star-  \Dot{\delta}_{0,\gamma}^\star  { ,}\\\frac{\partial \widehat{\tau}_{\MR}\left(\beta^\star, \gamma^\star, \theta^\star\right)}{\partial \theta^\T}=\frac{R}{\rho}\left( - \Dot{\delta}_{1,\theta}^\star + \Dot{\delta}_{0,\theta}^\star\right)+ \Dot{\delta}_{1,\theta}^\star-  \Dot{\delta}_{0,\theta}^\star { .}
  \end{gather*}
Under submodel  $\cup_{j=1}^3\mathcal{M}_j$,    we have 
\begin{align*}
\E\left\{\frac{\partial \widehat{\tau}_{\MR}\left(\beta^\star, \gamma^\star, \theta^\star\right)}{\partial \theta^\T}\right\}=0.
\end{align*} 
Therefore, 
\begin{align*}
\begin{aligned}
\widehat{\tau}_{\MR} (\widehat{\beta}, \widehat{\gamma}, \widehat{\theta})  &
\cong \widehat{\tau}_{\MR}\left(\beta^\star, \gamma^\star, \theta^\star\right) \\
& +n ^{-1} \sum_{ j } \E\left\{\frac{\partial \widehat{\tau}_{\MR}\left(\beta^\star, \gamma^\star, \theta^\star\right)}{\partial \beta }\right\} \E\left\{\frac{\partial S_\pi\left(Z, X, V; \beta^\star\right)}{\partial \beta }\right\}^{-1}{ \frac{R}{\rho}} S_\pi\left(Z_j, X_j, V_j ; \beta^\star\right) \\
& +n ^{-1} \sum_{ j } \E\left\{\frac{\partial \widehat{\tau}_{\MR}\left(\beta^\star, \gamma^\star, \theta^\star\right)}{\partial \gamma^\T}\right\} \E\left\{\frac{\partial S_\mu\left(Z, X, V, Y ; \gamma^\star\right)}{\partial \gamma^\T}\right\}^{-1} { \frac{R}{\rho}}S_\mu\left(Z_j, X_j, V_j, Y_j ; \gamma^\star\right) .
\end{aligned}
\end{align*} Thus, we   derive the inﬂuence function for the estimator  $\widehat\tau_{\MR}$ . 
\end{proof}
\subsection{Optimal sampling allocation}
Under the submodel $\mathcal{M}_1\cup\mathcal{M}_2$, the goal is to minimize \eqref{eq:var-tau-rand}. We first calculate the asymptotic variance, 
\begin{align*} 
{ \mathrm{asyvar}}\{\widehat\tau_\MR(\widehat{\beta}, \widehat{\gamma}, \widehat{\theta})\}&=\frac{\E\{\psi_{\MR}^{\tau}(O)\}^2}{n}\\&=  \frac{\E({S_2^2}  \mid R=1)}{n \rho
} +\frac{2 \E( {S_2 S_1}  \mid R=1) +\E\left(S_1^2\right) }{n} \\ &=  \frac{C_1+\rho C_2}{C}\left(  { \Gamma_1 }+\frac{ \Gamma_2  }{\rho} \right) \\&= \left(\frac{C_1}{C}+\frac{C_2}{C} \rho\right)\left(\Gamma_1+\frac{\Gamma_2}{p}\right) \\
 &=\frac{C_1\Gamma_1}{C} +\frac{C_1}{C} \frac{\Gamma_2}{\rho}+\frac{C_2}{C} \rho \Gamma_1+\frac{C_2 \Gamma_2}{C} \\
  &=\frac{C_1\Gamma_1}{C} +\frac{C_2 \Gamma_2}{C}+ \frac{C_1\Gamma_2}{C}\frac{1}{\rho}+\frac{C_2\Gamma_1}{C} \rho ,
\end{align*}
where the first equlity holds due to 
\begin{align*}
     \E\{\psi_{\MR}^{\tau}(O)\}^2 &= \E\left(\frac{RS_2}{\rho} +S_1\right)^2\\&=   \E\left(\frac{R}{\rho^2} S_2^2+2\frac{R}{e} S_2 S_1+S_1^2\right)  \\&= \frac{ \E({S_2^2}  \mid R=1)}{\rho} +{2 \E( {S_2 S_1}  \mid R=1) +\E\left(S_1^2\right) }{ ,}
\end{align*}
Let 
\begin{align*}
   \Gamma_1&=\E\left\{S_1^2(\beta^\star, \gamma^\star, \theta^\star)\right\} +2 \E\{ S_2(\beta^\star, \gamma^\star, \theta^\star) S_1(\beta^\star, \gamma^\star, \theta^\star)\mid R=1\}  ,\\
   \Gamma_2&=\E\left\{S_2^2(\beta^\star, \gamma^\star, \theta^\star)\mid R=1\right\}.
\end{align*} 
Then the optimal sample size allocation is,
\begin{align*} 
    \rho^\star= \begin{cases}\sqrt{\dfrac{\Gamma_2 C_1}{\Gamma_1 C_2}}, & \text { if } \Gamma_2 C_1<\Gamma_1 C_2, \\ 1, & \text { if } \Gamma_2 C_1\geq \Gamma_1 C_2,\end{cases}~~~   n^\star= \frac{C}{C_1+\rho^\star C_2}.
\end{align*}

 \section{The proof of Proposition \ref{PROP:CRR-IDENTIF}}
The proof for the three nonparametric identification strategies for $\xi$ follows a similar structure to the proof of Theorem \ref{THEOREM:IDENTIFIABLE}, and  we omit it for simplicity.
 \section{The proof of theorem \ref{THM: MTYPE-EIF}}

 \begin{proof}
     
We now consider the EIF for   $\xi=\tau_1/\tau_0$, where $\tau_z=\E(Y_z)$, so that 
\begin{align*}
      \varphi_{\mathrm{eff},z}(O)&= \frac R{\rho(X)}\left[\frac {\mathbb{I}(Z=z)\left\{Y-\mu(z,X,V)\right\}}{\{\pi(X,V)\}^2\{1-\pi(X,V)\}^{1-z}}+\mu(z,X,V)- \delta(z,X)\right]  + \delta(z,X)  -\tau_z.
\end{align*}
Let $  \varphi_{z}(O)=\varphi_{\mathrm{eff},z}(O)+\tau_z$. 
By the standard Gateaux derivative approach \citep{kennedy2023semiparametric},  we have
\begin{align*} 
 \mathbb{I} \mathbb{F}(\xi)&= \mathbb{I F}\left(\frac{\tau_{1}}{\tau_{0}}\right) \\&=\frac{\mathbb{I} \mathbb{F}\left(\tau_{1}\right)}{\tau_{0}}-\left(\frac{\tau_{1}}{\tau_{0}}\right) \frac{\mathbb{\mathbb { IF }}\left(\tau_{0}\right)}{\tau_{0}} \\&  = \frac{\varphi_{\mathrm{eff},1}(O) -\xi{\varphi_{\mathrm{eff},0} }(O) }{\tau_{0}}\\&=
 \frac{\varphi_{1}(O)-\tau_1}{\tau_0}-\xi\left\{\frac{\varphi_{0}(O)-\tau_{0}}{\tau_{0}}\right\}\\&=
 \frac{\varphi_{1}(O)-\xi\varphi_{0}(O) }{\tau_{0}}.
\end{align*}

 \end{proof}
 \section{The proof of Theorem \ref{THM: MTYPE-MUL-ROBUSTNESS}}
 The proofs for the    triple robustness property for $\widehat\xi_\MR$ are similar to the proof of Theorem \ref{THM: MUL-ROBUSTNESS} and hence omitted. 

  To   demonstrate its local efficiency, we consider the one-step estimator $$\widehat{\xi}_{0}=\Pn\{{\widehat{\varphi}_{1}(O)-\xi\widehat{\varphi}_{0}}(O)\}\big/\tau_0 +\xi.$$  
 Since $\widehat  \xi_{0}$ can achieve the local efficiency, and the second claim follows that 
 \begin{align*}
  \notag  \sqrt{n}(  \widehat \xi_{\MR}-\widehat \xi_{0})&=\sqrt{n} \Pn\{{\widehat\varphi_{1}(O)-\xi\widehat\varphi_{0} }(O)\}\left[\frac{1}{\Pn \{\widehat\varphi_{0} (O)\}} -\frac{1}{\tau_0}\right] =o_p(n^{-1/2}).
 \end{align*}
 \section{Additional simulation studies for a binary outcome}
  In this section, we investigate the ﬁnite-sample performance for  a 
  binary outcome.  We consider the following data-generating mechanism: 
 \begin{itemize} 
    \item[(a)] Observed covariate: $X\sim N(0,1)$.
    \item[(b)] Partially observed covariate: $V\mid X \sim N(1 + 0.75X,1)$.
    \item[(c)] Treatment assignment: $\pr(Z=1\mid X,V)=\mathrm{expit}(0.5+0.5X+0.5V)$.
    \item[(d)] Selection mechanism: $\pr(R=1\mid X)=\mathrm{expit}(0.75 + X )$.
    \item[(e)] Observed outcome: $\pr(Y=1 \mid Z,X,V)= \mathrm{expit}(-2+0.5X+0.5V+4Z+0.5ZX+0.5ZV)$. 
\end{itemize} We are interested in estimating the causal risk ratio $\xi$, with a true value of $4$. We also evaluate the   finite-sample properties of four semiparametric estimators: the first three estimators, as implied by Proposition \ref{PROP:CRR-IDENTIF}, which are analogous to the estimators $\widehat\tau_1$, $\widehat\tau_2$, and $\widehat\tau_3$; and the  triply robust estimator $\widehat{\xi}_{\MR}$.  The simulation results are presented in Table \ref{tab:binaryoutcome}. The  bias, empirical standard deviation, and the 95\% coverage probability are all multiplied by $10^2$.  
 \section{Application}
  \subsection{Demographic
Characteristic}
Table \ref{application}  provides summary statistics on the demographic characteristics in the primary dataset and linked cohort dataset. 
  	 \begin{landscape}
		\begin{table}
 \caption{Bias, standard deviation (SD), and 95\% coverage probability (CP) of the proposed estimators with a binary outcome under different sample sizes. The bias, SD and CP are multiplied by $10^2$.}
  \label{tab:binaryoutcome}
			\centering
			\begin{threeparttable}    
				\renewcommand{\arraystretch}{0.75}
				\resizebox{0.835\columnwidth}{!}{%
	\begin{tabular}{cccccccccccccccccccc}
    \toprule
    
    \specialrule{0em}{3pt}{3pt}
    
    \multicolumn{3}{c}{Model} & 
    \multicolumn{1}{c}{} & 
    \multicolumn{3}{c}{$\widehat\tau_{1}$} & 
    \multicolumn{1}{c}{} & 
    \multicolumn{3}{c}{$\widehat\tau_{2}$} & 
    \multicolumn{1}{c}{} & 
    \multicolumn{3}{c}{$\widehat\tau_{3}$}& 
    \multicolumn{1}{c}{} & 
    \multicolumn{3}{c}{$\widehat\tau_{\MR}$}  \\
    \specialrule{0em}{3pt}{3pt}
    
    \cmidrule(lr){1-3}	\cmidrule(lr){5-7} 	\cmidrule(lr){9-11}	\cmidrule(lr){13-15}	\cmidrule(lr){17-19}
    
    \multirow{2}{*}{$\mathcal{M}_1$} &	\multirow{2}{*}{$\mathcal{M}_2$} & 	\multirow{2}{*}{$\mathcal{M}_3$} 
    &\multicolumn{1}{c}{}
    &\multicolumn{3}{c}{Sample Size}
    &\multicolumn{1}{c}{}
    &\multicolumn{3}{c}{Sample Size} &\multicolumn{1}{c}{}
    &\multicolumn{3}{c}{Sample Size} &\multicolumn{1}{c}{}
    &\multicolumn{3}{c}{Sample Size} \\
    \cmidrule(lr){5-7} 	\cmidrule(lr){9-11}	\cmidrule(lr){13-15}	\cmidrule(lr){17-19}
    & &	 &&1000  & 5000 & 10000 	&& 1000  & 5000 & 10000	&& 1000  & 5000 & 10000 && 1000  & 5000 & 10000 \\
    \specialrule{0em}{0pt}{3pt}
    \midrule
    
    &Bias & &&&& &  & & & & & & & & & & &\\
    $\checkmark $& $\checkmark $& $\checkmark $  && 19 & 2 & 0 && 9 & 3 & 0 && 9 & 2 & 0 && 12 & 2 & 0 \\
    $\checkmark $ & $\times$& $\times$ && 19 & 2 & 0 && 91 & 83 & 78 && 68 & 60 & 55 && 13 & 2 & 0 \\
    $\times$ & $\checkmark $& $\times$ && 94 & 78 & 73 && 9 & 3 & 0 && 76 & 69 & 65 && 12 & 3 & 0 \\
    $\times$ & $\times$& $\checkmark $&& 97 & 83 & 78 && -35 & -40 & -43 && 9 & 2 & 0 && 11 & 3 & 0 \\
    $\times$ & $\times$& $\times$   && 97 & 83 & 78 && 67 & 59 & 54 && 68 & 60 & 55 && 60 & 51 & 46\\ \addlinespace[1mm]
    
    & {  SD} &  &  & & & & & & & & & & &&&&&\\
    $\checkmark $& $\checkmark $& $\checkmark $ && 81 & 34 & 23 && 65 & 28 & 21 && 65 & 28 & 21 && 71 & 29 & 21 \\
    $\checkmark $ & $\times$& $\times$ && 81 & 34 & 23 && 87 & 37 & 25 && 87 & 38 & 25 && 78 & 30 & 21 \\
    $\times$ & $\checkmark $& $\times$ && 95 & 38 & 25 && 65 & 28 & 21 && 79 & 39 & 27 && 70 & 30 & 21 \\
    $\times$ & $\times$& $\checkmark $  && 93 & 38 & 25 && 60 & 27 & 19 && 65 & 28 & 21 && 68 & 29 & 21 \\
    $\times$ & $\times$& $\times$  && 93 & 38 & 25 && 87 & 38 & 26 && 87 & 38 & 25 && 86 & 36 & 25  \\\addlinespace[1mm]
    
    & { CP} &  &  & & & & & & & & & & &&&&&\\
    $\checkmark $& $\checkmark $& $\checkmark $ && 97  & 95 & 96 && 97  & 95 & 93 && 97  & 95  & 94 && 95 & 95  & 95 \\
    $\checkmark $ & $\times$& $\times$ && 97 & 95 & 96 && 99 & 41 & 65 && 99 & 71 & 44 && 96 & 96 & 97 \\
    $\times$ & $\checkmark $& $\times$ && 98 & 51 & 16 && 97 & 95 & 93 && 100 & 67 & 41 && 97 & 95 & 96 \\
    $\times$ & $\times$& $\checkmark $&& 99 & 44 & 80 && 85 & 64 & 33 && 97 & 95 & 94 && 97 & 95 & 96 \\
    $\times$ & $\times$& $\times$&  & 99 & 44 & 80 && 99 & 71 & 48 && 99 & 71 & 44 && 99 & 81 & 61\\
    
    \bottomrule
\end{tabular}
} 
			\end{threeparttable}       
		\end{table}
	\end{landscape}

	\begin{landscape}
		\centering
			\begin{ThreePartTable}
		
				\begin{TableNotes}
					\small              
				\item[1] Values are expressed as median (interquartile range)     
				\item[2] Values are expressed as median (interquartile range)     
			\end{TableNotes} 
            \setlength\tabcolsep{1pt}{
			\begin{longtable}{lcccccccccccc}
		\caption{Summary statistics of the demographic characteristics
between the primary data and the linked data.}\label{application} \\
				\toprule
				\multirow{3}*{} &\multicolumn{1}{c}{\quad\quad\quad}& \multicolumn{4}{c}{Physical activity} & \multicolumn{2}{c}{\quad\quad\quad}&\multicolumn{4}{c}{No physical activity}    \\
				\cmidrule(lr){3-6}\cmidrule(lr){9-12} 
                 &\multicolumn{1}{c}{}
				& \multicolumn{2}{c}{Primary data(9174)} &\multicolumn{2}{c}{Linked data(2661)}&\multicolumn{2}{c}{}
				&\multicolumn{2}{c}{Primary data(9600)}
				&\multicolumn{2}{c}{Linked data(2565)}\\
				\cmidrule(lr){3-4}\cmidrule(lr){5-6} \cmidrule(lr){9-10} \cmidrule(lr){11-12} 
				
				&&No.  & $\%$ &No.  & $\%$& & &No.  & $\%$ &No.  & $\%$  \\
				
				\midrule
				\endfirsthead 
				\caption{Summary statistics of the demographic characteristics
between the primary data and the linked  data.}\\
					\toprule
				\multirow{3}*{} &\multicolumn{1}{c}{\quad\quad\quad}& \multicolumn{4}{c}{Physical activity} & \multicolumn{2}{c}{\quad\quad\quad}&\multicolumn{4}{c}{No physical activity}    \\
				\cmidrule(lr){3-6}\cmidrule(lr){9-12} 
                 &\multicolumn{1}{c}{}
				& \multicolumn{2}{c}{Primary data(9174)} &\multicolumn{2}{c}{Linked data(2661)}&\multicolumn{2}{c}{}
				&\multicolumn{2}{c}{Primary data(9600)}
				&\multicolumn{2}{c}{Linked data(2565)}\\
				\cmidrule(lr){3-4}\cmidrule(lr){5-6} \cmidrule(lr){9-10} \cmidrule(lr){11-12} 
				
				&&No.  & $\%$ &No.  & $\%$& & &No.  & $\%$ &No.  & $\%$  \\

				\midrule
	\endhead
\midrule
\endfoot
\bottomrule
\insertTableNotes         
\endlastfoot

				Age\tnote{1} && 49(30) &  &40(25)  && &&51(29) && 41(24) & \\
				Male && 4670  & 50.9  & 1385  & 52.0 & & & 3987   &41.5   & 1016   & 39.6   \\
				Race && & & & & & & &  & & \\
				\quad White& &7090 &77.3 &2032 &76.4 &&& 7156 &74.6 & 1888 &73.6\\
				\quad Black& &1213  &13.2  &380 &14.3 & && 1585  &16.5  &471 &18.4\\
	         	\quad Amer Indian/Hawaiian/Alaska Native& &80  &0.9  &37  &1.4 & & & 68 &0.7   &27  &1.1\\
				\quad Asian/NATV Hawaiian/PACFC ISL &&515  &5.6  &113  &4.2 && & 512  &5.3  &101  &3.9\\
				\quad Multiple races &&276  &3.0  &99  &3.7 &&  &279  &2.9  &78  &3.0\\
				Marriage status && & & && & & & &  &  \\
				\quad Married &&4743 &51.7  &1111 &45.5& && 4839 &50.4 & 1207 &47.1\\
				\quad Widowed& &568  &6.2  &60 &2.3 & && 810  &8.4  &66 &2.6\\
				\quad Divorced &&1198  &13.1 &387  &14.5 && &1223 &12.7   &326  &12.7\\
				\quad Separated &&205  &2.2  &60  &2.3 && &263  &2.8  &73  &2.8\\
				\quad Never married &&2460 &26.8  &943  &35.4 & & &2465  &25.7 &893  &34.8\\
				BMI\tnote{2} &&26.3(7.6) & & 26.7(8.1)& &&& 27.4(9.6)& & 28.1(10.0)& \\
				Education &&& & & & & & & & &   \\
				\quad No degree &&1170 &12.8 &302 &11.3 &&& 1544 &16.1 & 383 &14.9\\
				\quad General educational development &&404  &4.4  &162 &6.1  & &&381  &4.0  &113 &4.4\\
				\quad High school diploma &&3955  &43.1 &1373  &51.6& & &4044 &42.1   &1275  &49.7\\
				\quad Bachelor's degree &&1772 &19.3 &413  &15.5& & &1702 &17.7   &384  &15.0\\
				\quad Master's degree &&803 &8.8  &76  &2.9 & & &828  &8.6 &100 &3.9\\
				\quad Doctorate degree &&212 &2.3  &11  &0.4& & &200  &2.1 &12  &0.5\\
				\quad Other degree &&858 &9.3  &324  &12.2& & &901  &9.4 &298  &11.6\\\\
				Poverty of family &&& && & & & & & &    \\
				\quad Poor/negative &&1222 &13.3 &235 &8.8 &&& 1516 &15.8 & 247 &9.6\\
				\quad Near poor &&379  &4.1 &108 &4.1&  && 496 &5.2  &131 &5.1\\
				\quad Low income &&1208  &13.2 &444 &16.7& & &1400 &14.6  &441 &17.2\\
				\quad Middle income &&2653  &28.9  &976 &36.7& & &2755 &28.7  &956 &37.3\\
				\quad High income &&3712 &40.5 &898  &33.7& & &3433 &35.7 &790  &30.8\\
				Smoking && 1765&19.2 &678 &25.5 &&& 1798& 18.7& 556&21.7\\
					Cancer diagnosis && 1031&11.2 &164 &6.2 &&& 1171& 12.2& 149&5.8\\
						Health insurance coverage && 8404&91.6 &2408 &90.5 &&& 8765& 91.3& 2285&89.1\\ 
		\end{longtable}}
		
		\end{ThreePartTable}
	\end{landscape} 
\end{document}